\documentclass[11pt,3p]{elsarticle}
\makeatletter
\def\ps@pprintTitle{\let\@oddhead\@empty
  \let\@evenhead\@empty
  \def\@oddfoot{\reset@font\hfil\thepage\hfil}
  \let\@evenfoot\@oddfoot
}
\makeatother

\usepackage{amsmath}
\usepackage{amsfonts}
\usepackage{amssymb}
\usepackage{bbm}
\usepackage{amsfonts}
\usepackage{amsthm}
\usepackage{orcidlink}
\newtheorem{theorem}{Theorem}[section]
\newtheorem{proposition}{Proposition}[section]
\usepackage{amscd}
\usepackage{mathtools}
\usepackage{mathrsfs}
\usepackage{commath}
\usepackage{lmodern}
\usepackage{color}
\usepackage[dvipsnames]{xcolor}
\usepackage{a4wide}
\usepackage{bm} 
\usepackage{physics}
\usepackage{enumitem}  
\usepackage{bbm}
\usepackage{hyperref}
\hypersetup{ colorlinks=true, urlcolor  = blue, linkcolor = blue,
citecolor = blue1,}

\usepackage[english]{babel}
\usepackage[IL2]{fontenc}
\usepackage[utf8]{inputenc}

\usepackage[linesnumbered,ruled,vlined]{algorithm2e}
\usepackage[noend]{algpseudocode}
\SetKwInput{KwInput}{Input}                
\SetKwInput{KwOutput}{Output}              

\SetCommentSty{mycommfont}

\newtheorem{lemma}[theorem]{Lemma}

\theoremstyle{definition}

\theoremstyle{remark}
\theoremstyle{remark}

\newcommand{\FP}{{\mbox{\rm\tiny FP}}}
\newcommand{\LC}{{\mbox{\rm\tiny LC}}}
\newcommand{\SP}{{\mbox{\rm\tiny SP}}}

\begin{document}

\begin{frontmatter}

\title{Analysis of inverse stochastic resonance: Effects of neural excitability and timescale separation}
\author[1,2]{Marius E. Yamakou \orcidlink{0000-0002-2809-1739}} 
\ead{marius.yamakou@fau.de}
 \author[2]{Torben Kr\"uger}
 \author[2]{Hermann Schulz-Baldes}
 \address[1]{Department of Data Science, Friedrich-Alexander-Universit\"at Erlangen-N\"urnberg, N\"urnberger Str. 74, 91052 Erlangen, Germany}
\address[2]{Department of Mathematics, Friedrich-Alexander-Universit\"at Erlangen-N\"urnberg, Cauer Str. 11, 91058 Erlangen, Germany}

\begin{abstract}
We analyze inverse stochastic resonance (ISR) in a bistable FitzHugh--Nagumo
neuron driven by additive noise in the voltage variable, focusing on how neural excitability and timescale separation regulate the noise-induced modulation of spiking activity. A codimension-two bifurcation analysis identifies a narrow bistable region in
which a stable fixed point and a stable limit cycle coexist, separated by an
unstable periodic orbit. Finite-time Monte Carlo simulations show that the occupation of the
limit-cycle basin may appear to depend on the initial basin when rare
transitions are not fully resolved. We prove that this dependence is not
asymptotic: the stochastic system admits a unique invariant probability
measure, so long-time firing statistics are independent of the initial
basin of attraction. The parameter dependence of ISR is characterized by quasi-potential barriers
computed with a geometric minimum action method for the degenerate noise. The difference
between the limit-cycle and fixed-point quasi-potentials partitions the
bistable wedge into two escape-dominated regimes. A reduced metastable
two-state Markov approximation yields a weak-noise formula for the limit-cycle basin
occupation probability, a sign criterion for genuine ISR, and a semiquantitative
prediction of the ISR-minimizing noise amplitude. In the fitted regime with a
negative effective exponent, a genuine ISR minimum occurs only when the
limit-cycle quasi-potential exceeds that of the fixed point. These
results provide an escape-balance mechanism linking intrinsic neuronal
parameters to asymptotic noise-induced spike suppression.
\end{abstract}

\begin{keyword}
biological neurons, excitability, timescale separation, bistability, quasi-potential, invariant probability measure, inverse stochastic resonance
\end{keyword}
\end{frontmatter}

\section{Introduction}
\label{sec:introduction}
Deterministic nonlinear systems can exhibit fixed points, periodic orbits, and chaotic trajectories. Multistability—coexistence of multiple attractors separated by unstable states---is common. In neurons, this manifests as coexisting quiescent (fixed point) and spiking (limit cycle) states. The final deterministic state of the system depends on the initial basin of attraction, but the presence of noise can induce transitions between basins. 

\textit{Inverse stochastic resonance} (ISR) is a noise-induced phenomenon that can occur in bi-stable stochastic dynamical systems where a stable fixed point and a stable limit cycle coexist. Weak noise suppresses oscillations along the limit cycle, minimizing the response measure (sometimes to zero); beyond this point, increased noise amplitude steadily raises the response \cite{tuckwellAnalysisInverseStochastic2012,yeInverseStochasticResonance2023,luInverseStochasticResonance2020,gutkinInhibitionRhythmicNeural2009,tuckwellInhibitionModulationRhythmic2009,guoInhibitionRhythmicSpiking2011,yamakouInverseStochasticResonance2024}. ISR occurs when some non-zero noise amplitude maximizes residence time in the fixed point’s basin of attraction, reducing transitions to the limit cycle’s basin at a particular noise amplitude. It is worth emphasizing that the ISR effect differs fundamentally from other noise-induced resonance phenomena, including stochastic resonance (SR) \cite{longtin1993stochastic,berglund2005universality,berglund2006noise,zamani2020concomitance,gammaitoni1998stochastic}, coherence resonance (CR) \cite{pikovskyCoherenceResonanceNoisedriven1997,yamakou2023combined,yamakouCoherenceResonanceStochastic2023}, and self-induced stochastic resonance (SISR) \cite{muratov2005self,devilleSelfinducedStochasticResonance2007,yamakouControlCoherenceResonance2019,zhu2024reduced}.

From the perspective of stochastic dynamical systems, ISR has been shown to emerge through several distinct mechanisms. One such mechanism occurs in a bistable parameter regime bounded by a subcritical Hopf bifurcation (HB) and a saddle-node-of-limit-cycles (SNLC) bifurcation, where noise-induced asymmetric switching between oscillatory and quasi-stationary states produces the characteristic suppression of spiking activity
\cite{gutkinInhibitionRhythmicNeural2009,tuckwellInhibitionModulationRhythmic2009}. Intermediate noise causes asymmetric switching, favoring the quasi-stationary state \cite{Torres2020theoretical,uzuntarlaDoubleInverseStochastic2017}, creating ISR's characteristic non-monotonic frequency-noise relationship. In this paper, we focus on bistability-driven ISR, though other mechanisms exist \cite{bacicTwoParadigmaticScenarios2020}, limit cycle phase-sensitivity \cite{Igor2018phase}, and noise-stabilized unstable fixed points \cite{bacicTwoParadigmaticScenarios2020,zhuUnifiedMechanismInverse2021}, which do not require bistability.

ISR was first observed in neurons, where optimal noise suppresses or halts firing \cite{gutkinInhibitionRhythmicNeural2009}. Validated in biological \cite{buchinInverseStochasticResonance2016} and physical systems \cite{huhControlStochasticInverse2023}, ISR demonstrates the noise's role in balancing excitation and inhibition. While potentially limiting spiking-based information transfer in neural networks, ISR enables computations requiring reduced activity without chemical inhibition \cite{paydarfarNoisyInputsInduction2006}. It may prevent pathological memory retention, regulate brain states, and offer therapeutic options for epilepsy \cite{buchin2015modeling}. ISR also impacts liquid crystal electroconvection \cite{huhInverseStochasticResonance2016}, ecological dynamics \cite{touboulComplexDynamicsSavanna2018}, airfoil system \cite{Zhu2025}, and excitable active rotators with adaptive coupling \cite{Bacic2018}.

In this paper, we investigate ISR in a neuronal model. Numerous studies have numerically investigated ISR in neurons and networks, examining factors including: noise types~\cite{luInverseStochasticResonance2020,wangNonGaussianNoiseAutapseinduced2022,zhaoLevyNoiseinducedInverse2019}, spatial extension~\cite{tuckwellEffectsVariousSpatial2011,liuEffectsNeuronalMorphology2024,zhangAutapseinducedMultipleInverse2021}, electromagnetic induction~\cite{yeInverseStochasticResonance2023}, conductance inputs~\cite{tuckwellInhibitionModulationRhythmic2009}, morphology~\cite{liuEffectsNeuronalMorphology2024}, time delays and coupling~\cite{liuEffectsNeuronalMorphology2024,zhangAutapseinducedMultipleInverse2021}, synapse types~\cite{uzuntarlaDoubleInverseStochastic2017}, autapses~\cite{zhangAutapseinducedMultipleInverse2021}, scale-free connectivity~\cite{uzuntarlaDoubleInverseStochastic2017}, and network adaptivity~\cite{yamakouInverseStochasticResonance2024,uzuntarlaDoubleInverseStochastic2017,bacicTwoParadigmaticScenarios2020,bacicInverseStochasticResonance2018}.

However, a detailed analytical study of how neural excitability and timescale separation jointly affect ISR remains lacking. Neural excitability—the ease with which deterministic neurons fire action potentials—varies across neuron types due to factors including: (i) ion channel expression (Na$^+$, K$^+$, Ca$^{2+}$ channel densities) and (ii) membrane properties (capacitance, resistance) affecting signal response \cite{xu2014parameters}. Timescale separation between membrane potential and recovery current is equally crucial, governing sodium channel kinetics and action potential morphology \cite{xu2014parameters}.

The main objective of this paper is to analyze how neural excitability and
timescale separation jointly regulate ISR in a bistable neuron. Using the FHN neuron model, we first identify the codimension-two bistable wedge in which a stable
fixed point and a stable limit cycle coexist and are separated by an unstable
periodic orbit. We then show, through Monte Carlo simulations, that finite-time
occupation curves of the limit-cycle basin may exhibit an apparent dependence on the initial basin when
one of the two escape processes is not sufficiently sampled. This dependence is
shown to be non-asymptotic by proving the existence and uniqueness of an
invariant probability measure for the stochastic system.

To characterize the parameter dependence of genuine ISR, we compute the
quasi-potential barriers for escape from the fixed-point and limit-cycle basins
using a degenerate-noise geometric minimum action method. The resulting barrier
difference partitions the bistable wedge into two escape-dominated regimes. By
combining weak-noise asymptotics for the mean escape times with a reduced
two-state Markov description of basin occupation, we derive a sign criterion for
genuine asymptotic ISR and a semiquantitative prediction of the noise amplitude
at which the ISR minimum occurs. This provides an escape-balance mechanism
linking the intrinsic neuronal parameters to long-time noise-induced suppression
of spiking.

The paper is organized as follows. Section~\ref{sec:mathematical model}
introduces the stochastic FHN neuron model. Section~\ref{sec:Zero-noise bifurcation}
identifies the bistable parameter regime through deterministic bifurcation
analysis. Section~\ref{sec: stochastic analysis} presents Monte Carlo
simulations illustrating the effects of excitability, timescale separation, and
initial basin on finite-time ISR curves. Section~\ref{invariant_measure} proves
existence and uniqueness of the invariant probability measure and establishes
the asymptotic independence of the long-time occupation statistics from the
initial condition. Section~\ref{sec:stochastic_analysis} computes the
quasi-potential barriers associated with escape from the fixed-point and
limit-cycle basins. Section~\ref{sec:_heuristics2} develops the reduced
escape-balance theory and derives a semiquantitative prediction for the
ISR-minimizing noise amplitude. Section~\ref{sec:conclusion} summarizes the
main conclusions.

\section{Mathematical model of the stochastic neuron}
\label{sec:mathematical model}
As a paradigmatic model with established neurobiological relevance, we consider a single FitzHugh-Nagumo (FHN) neuron~\cite{Fitzhugh-1960a,fitzhughImpulsesPhysiologicalStates1961,nagumoActivePulseTransmission1962}. This neuron is modeled by slow-fast stochastic differential equations (SDEs), represented on the fast timescale $t$ as in Eq.~\eqref{eq:2}:
\begin{equation}\label{eq:2}
\begin{split}
\left\{
\begin{array}{lcl}
dv_{t} &=& f(v_{t}, w_{t}) dt + \sigma N dB_t, \\[1.0mm]
dw_{t} &=& g(v_{t}, w_{t}) dt,
\end{array}
\right.
\end{split}
\end{equation}
where for any point $x=(v,w)$, the drift \(F(x)=\bigl[f(v,w),\,g(v,w)\bigr]^{\top}\) is given by
\begin{equation}\label{eq:3}
\begin{split}
\left\{
\begin{array}{lcl}
f(v, w) &=& v(a-v)(v-1)-w,\\
g(v, w) &=& \varepsilon(bv-cw),
\end{array}
\right.
\end{split}
\end{equation}
in which \( v_t \in \mathbb{R} \) and \( w_t \in \mathbb{R} \) represent the fast membrane potential and slow recovery current variables of the model at time \( t \), respectively, and \( b > 0 \), \( c > 0 \) are constant parameters. The noise enters only through the voltage variable. The corresponding diffusion vector is
\begin{equation}\label{diffusion}
N = (1,0)^\top ,
\end{equation}
where \(B_t\) denotes a standard one-dimensional Brownian motion, and \(\sigma\ge 0\) is the noise amplitude.
 
The timescale separation parameter $\varepsilon$ of the model, usually defined  $0<\varepsilon:=\tau/t\ll1$, represents the timescale ratio between slow ($\tau$) and fast ($t$) times dynamics, causing the membrane potential $v$ to evolve faster than the recovery current $w$. Biologically, $\varepsilon$ reflects sodium channel kinetics and shapes action potential morphology \cite{xu2014parameters}. 

The excitability parameter $a$ critically influences fast dynamics (particularly sodium current) and is typically set to $a > 0$ ($0 < a < 0.5$ or $0 < a < 1$ \cite{xu2014parameters}). However, our bifurcation analysis reveals $a < 0$ also has electrophysiological significance, as it enables bistability (co-existing stable fixed point and stable limit cycle), crucial for ISR in the FHN model.

Figures \ref{fig:2}(a)-(c) illustrate stochastic dynamics (\(\sigma > 0\)) of the model, obtained by integrating Eq.~\eqref{eq:2} in the bi-stable regime shown in Fig. \ref{fig:1}(a), using the Euler–Maruyama algorithm \cite{highamAlgorithmicIntroductionNumerical2001}. A spike is recorded whenever $v$ reaches or exceeds \(v_{\mathrm{th}} = 0.25\), since trajectories crossing this threshold autonomously make a large excursion in phase space before returning.

Figure~\ref{fig:2}(a) shows 29 periodic spikes for weak noise ($\sigma = 5.0 \times 10^{-4}$) during 2000 time units, after a thermalization of 5000 time units.  Surprisingly, a stronger noise amplitude ($\sigma = 5.0 \times 10^{-3}$) in Fig.~\ref{fig:2}(b) shows a complete inhibition of spikes. A stronger noise ($\sigma = 5.0 \times 10^{-2}$) in Fig.~\ref{fig:2}(c) yields only 22 spikes during the same time interval. 

This non-monotonic dependence of the spike count on the noise amplitude, characterized by maximal inhibition at an intermediate value of $\sigma$ despite identical initial conditions in the limit-cycle basin, defines ISR. Mechanistically, it arises from the asymmetry between the mean first exit times (MFETs) from the basins of attraction of the coexisting stable fixed point and stable limit cycle \cite{tuckwellAnalysisInverseStochastic2012,yeInverseStochasticResonance2023,luInverseStochasticResonance2020,gutkinInhibitionRhythmicNeural2009,tuckwellInhibitionModulationRhythmic2009}, as computed in Figs.~\ref{fig:2}(d) and \ref{fig:2}(e).

\begin{figure}
\centering
\includegraphics[width=5.0cm,height=4.0cm]{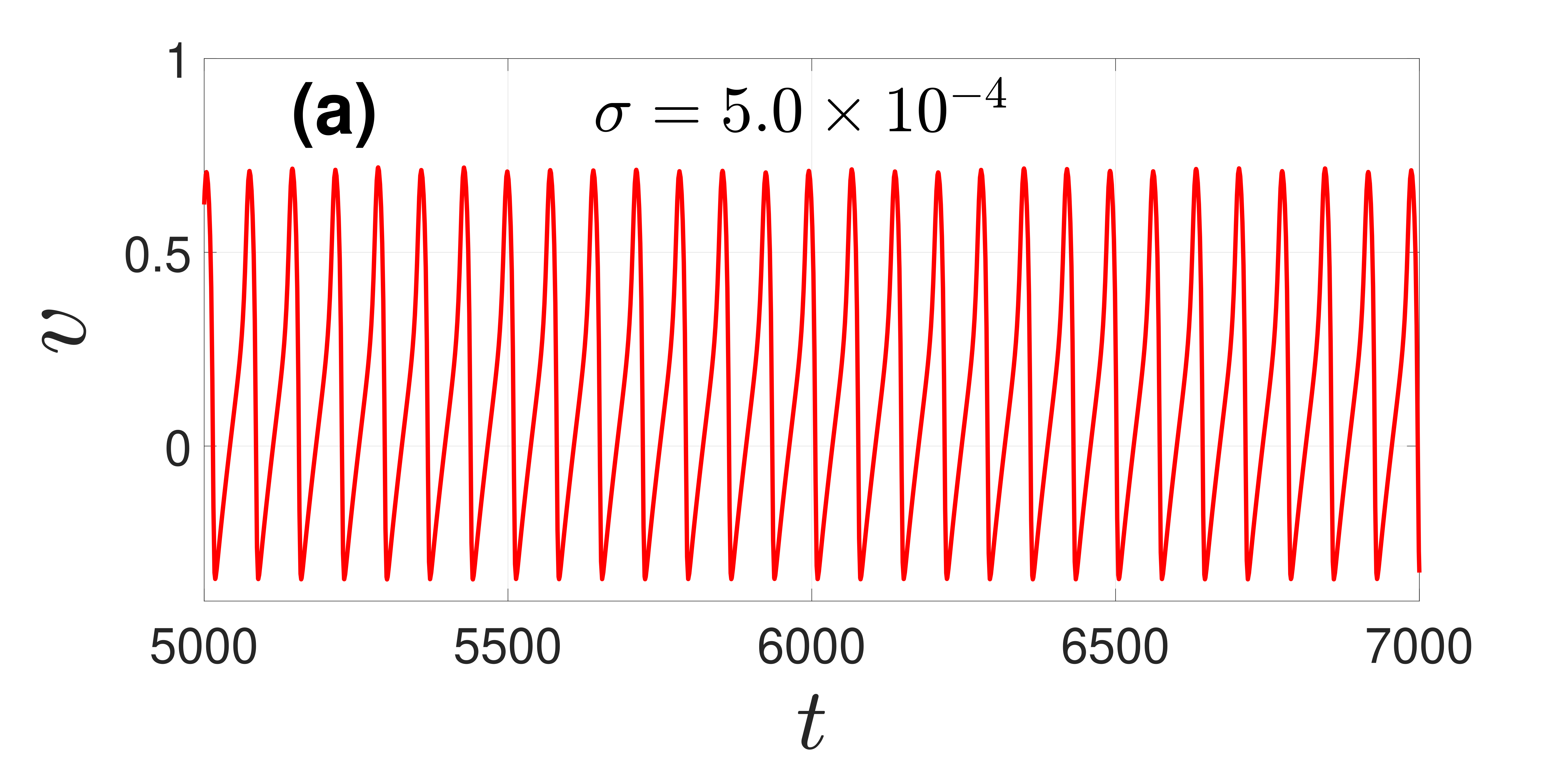}\includegraphics[width=5.0cm,height=4.0cm]{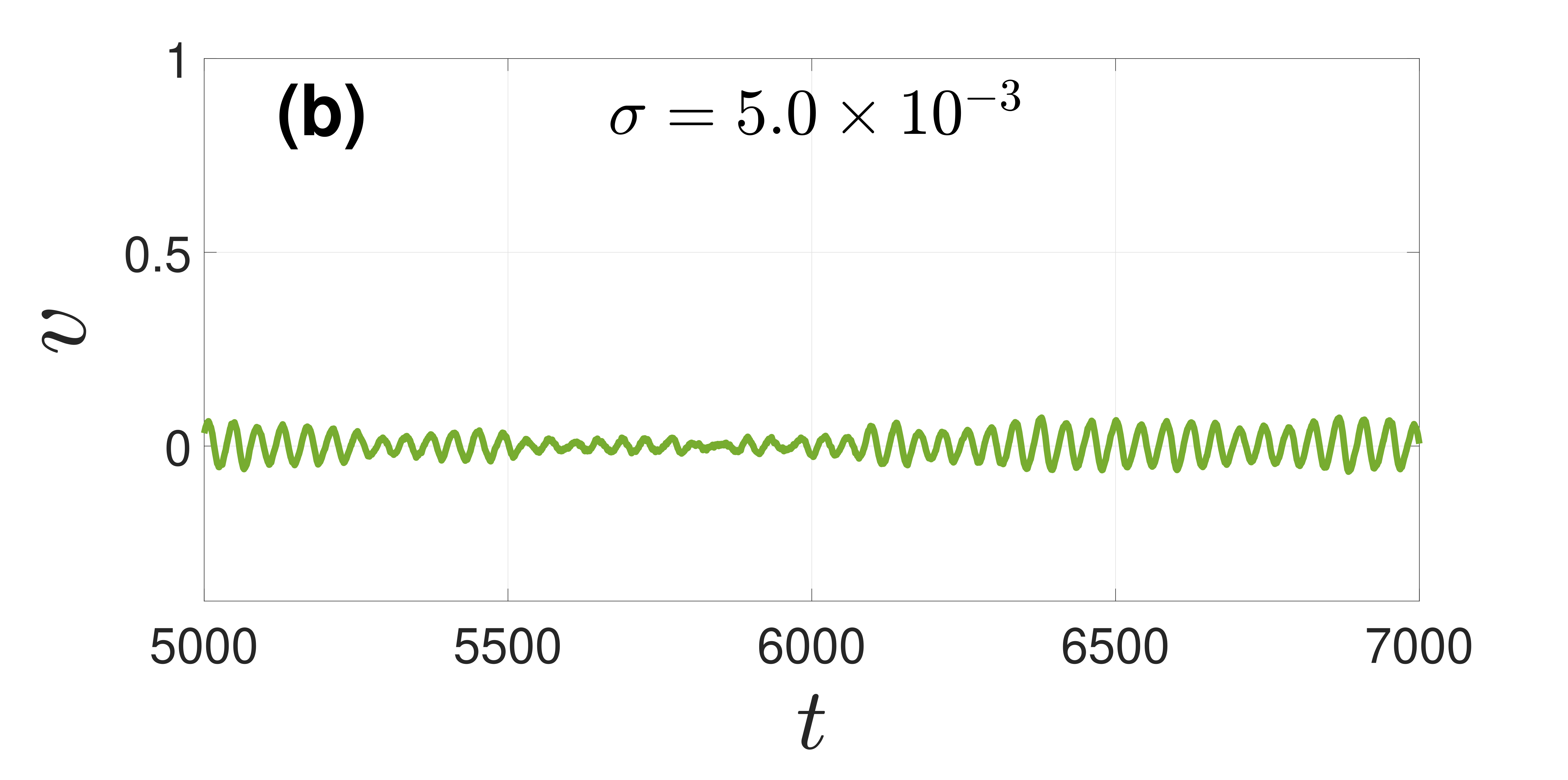}\includegraphics[width=5.0cm,height=4.0cm]{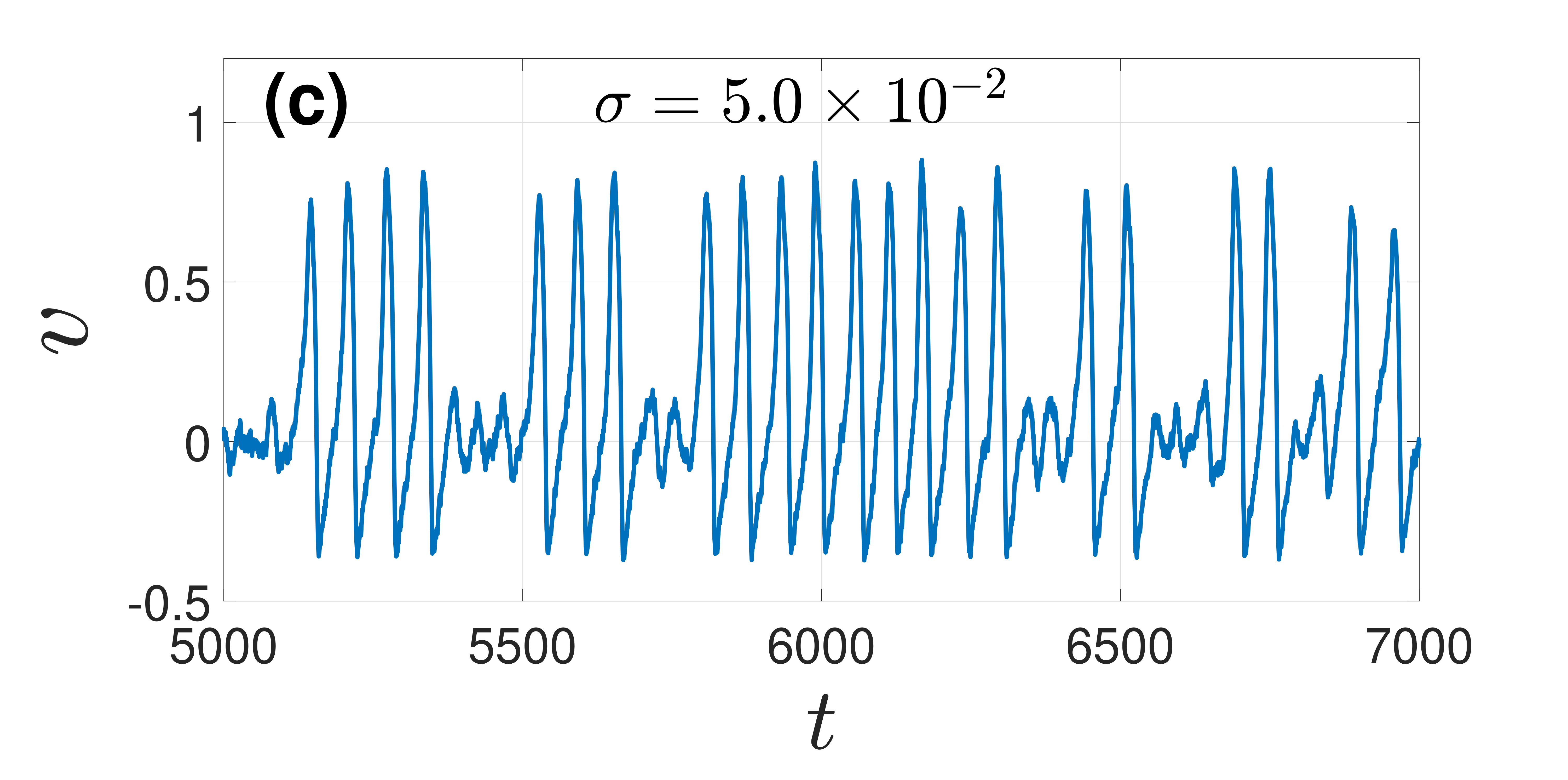}
\includegraphics[width=7.5cm,height=4.0cm]{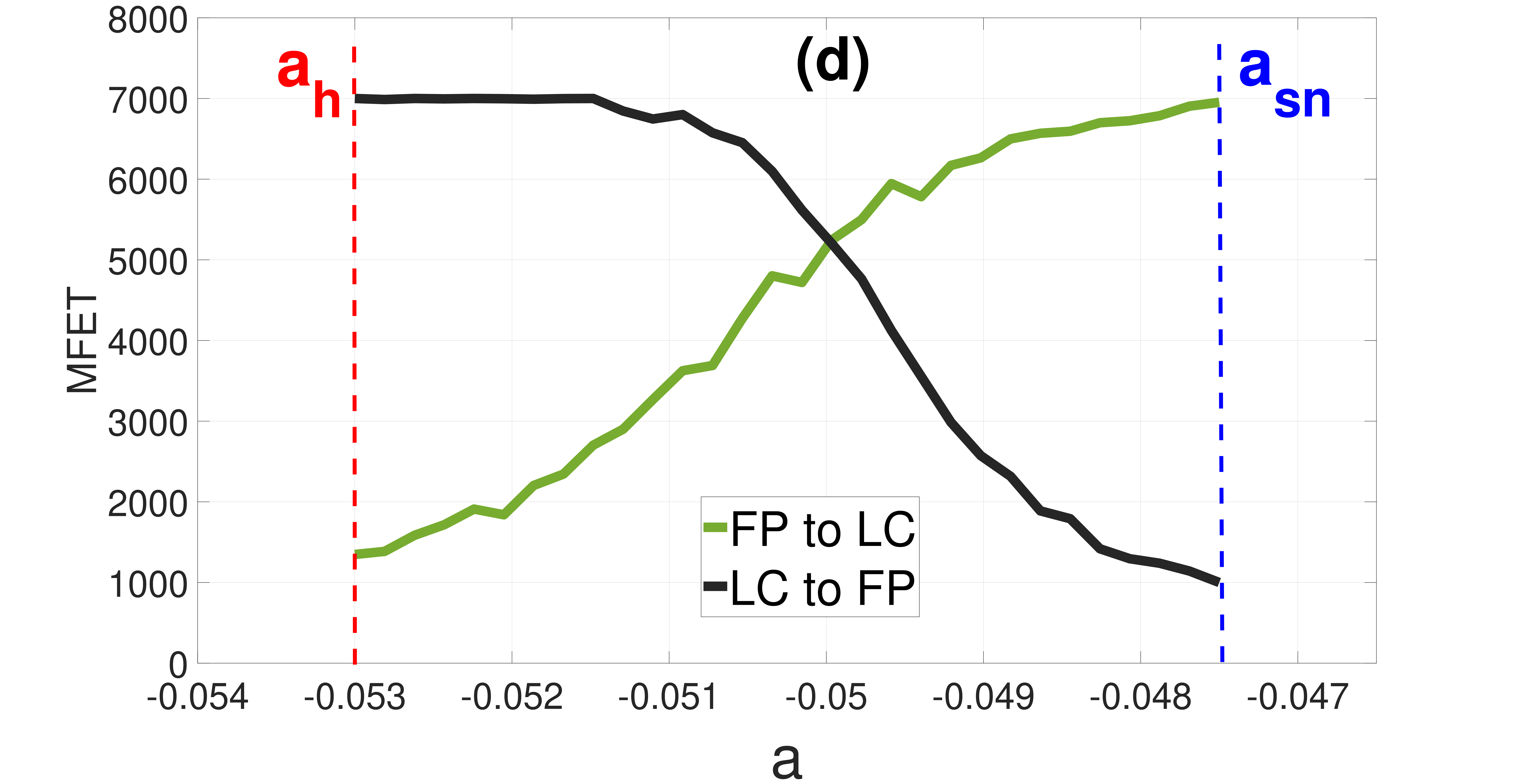}\includegraphics[width=7.5cm,height=4.0cm]{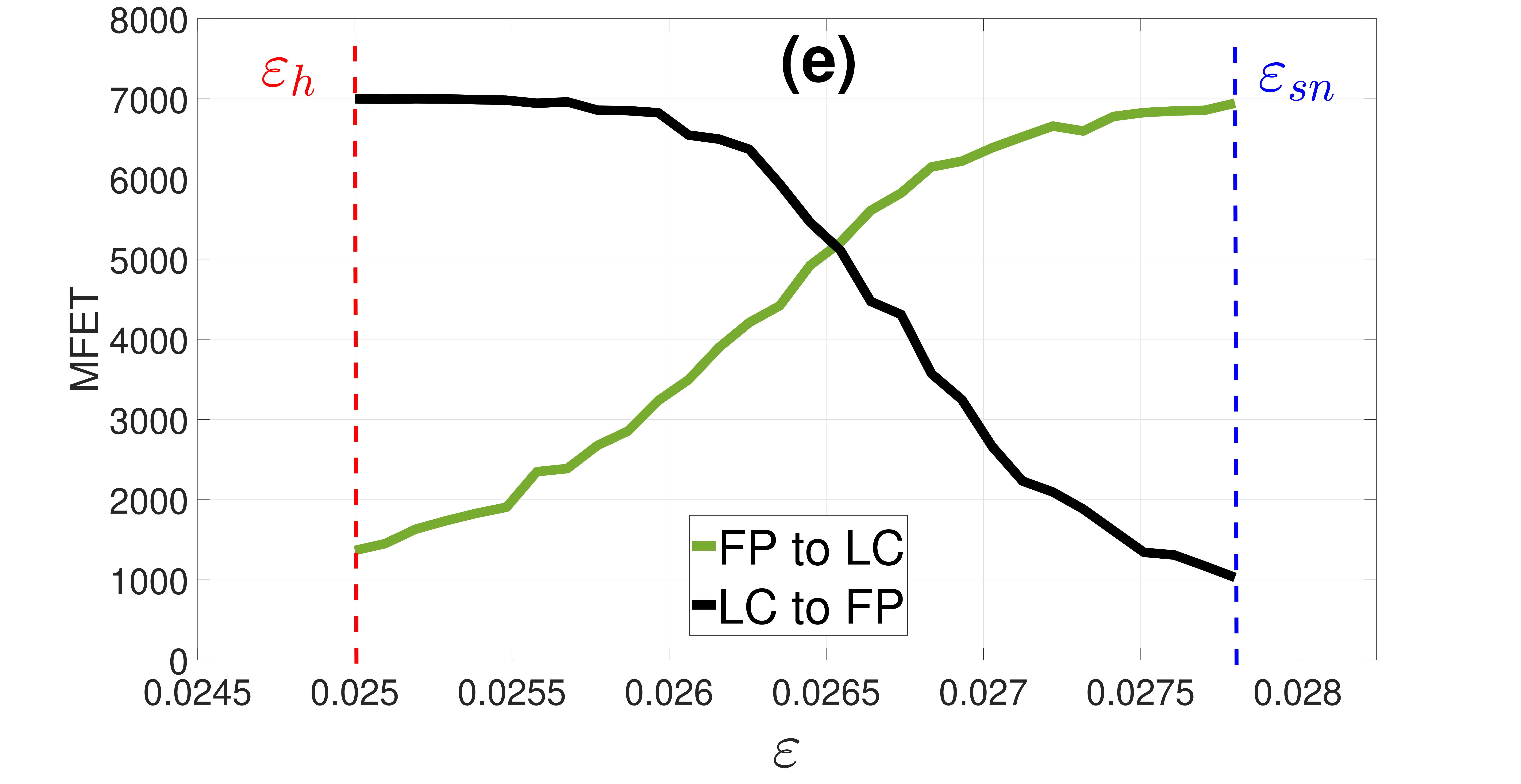}
\caption{(a)–(c) Time series of the potential variable \(v\) in Eq. \ref{eq:2} showing the non-monotonic effects on the number of spikes as noise amplitude  \(\sigma\) increases and when the initial conditions are \(X_0=(v(0),w(0)) = (0.75, 1.0)\). 
(d)–(e) MFET vs. $a$ (d) and $\varepsilon$ (e) within their bistability intervals. Green: MFET from $\mathcal{B}_\FP$ to $\mathcal{B}_\LC$; black: reverse. Vertical red/blue lines: Hopf/SNLC thresholds. $\varepsilon=0.0265$ in (d), $a=-0.05$ in (e), \(b = 1.0\), \(c = 2.0\) in all.}
\label{fig:2}
\end{figure}

\section{Zero-noise bifurcation analysis and bistability necessary for ISR}
\label{sec:Zero-noise bifurcation}
Dynamical systems theory establishes that bistability (coexisting stable fixed point and stable limit cycle) in the deterministic ($\sigma=0$) version of the system is necessary for ISR when the noise term is switched on ($0<\sigma\ll1$) \cite{tuckwellAnalysisInverseStochastic2012,yeInverseStochasticResonance2023,luInverseStochasticResonance2020,gutkinInhibitionRhythmicNeural2009,guoInhibitionRhythmicSpiking2011}. Therefore, we present a complete deterministic ($\sigma=0$) bifurcation analysis of the FHN model, to identify the precise parameter intervals of $a$ and $\varepsilon$ in which the model is bi-stable.

We carry out this deterministic bifurcation analysis on the slow timescale version of Eq. \eqref{eq:2}. That is, using the scaling  $0<\varepsilon:=\tau/t$, we rewrite the deterministic part of Eq. \eqref{eq:2} as
\begin{align}\label{eq:4}
\begin{split}
\left\{\begin{array}{lcl}
\varepsilon dV_{\tau}&=&[V_{\tau}(a-V_{\tau})(V_{\tau}-1)-W_{\tau}]d\tau,\\
dW_{\tau}&=&(bV_{\tau}-cW_{\tau})d\tau.
\end{array}\right.
\end{split}
\end{align}
The deterministic ODEs in Eq. \eqref{eq:2} ($\sigma=0$) and Eq. \eqref{eq:4} share identical topological phase space structure, preserving attractors' existence and trajectories (up to a time constant) \cite{kuehn2015multiple}. In the bifurcation analysis of Eq. \eqref{eq:4} below, we drop the subscripts $\tau$ on $V_{\tau}$ and $W_{\tau}$ for simplicity. The set of fixed points of Eq.~\eqref{eq:4}, defined by the intersection of the $V$- and $W$-nullclines, is given by 
\begin{equation}\label{eq:5}
\Big\{(V_e,W_e)\in\mathbb{R}^2:V_e(a-V_e)(V_e-1)-W_e=(bV_e-cW_e)=0\Big\},
\end{equation}
corresponds to the rest (quiescent) states of the neuron, where the variables $V$ and $W$ reach a stationary state. 
One obtains, from Eq.\eqref{eq:5}, the fixed point equations
\begin{align}\label{eq:6}
\begin{split}
\left\{\begin{array}{lcl}
 bV_e/c&=&-V_e^3+(a+1)V_e^2-aV_e,\\
  W_e&=&bV_e/c,
 \end{array}\right.
\end{split}
\end{align}
which have three solutions given as
\begin{align}\label{eq:7}
\begin{split}
\left\{\begin{array}{lcl}
V_1&=&0,\\[1.0mm]
V_{2,3}&=&(a+1)/2\pm \sqrt{(a-1)^2/4 -b/c}.
 \end{array}\right.
\end{split}
\end{align}
\begin{proposition}\label{proposition 1} If $(a-1)^2/4 < b/c$, then the origin $(V_1,W_1)=(0,0)$ is the unique fixed point of Eq. \eqref{eq:4}.
\end{proposition}

The linearization of Eq.\eqref{eq:4} at a given fixed point
$V_e$ is
\begin{equation}\label{eq:12}
\begin{split}
\left\{\begin{array}{lcl}
\varepsilon d\overline{V}&=& (-3V_e^2\overline{V} +2(a+1)V_e \overline{V} -a\overline{V} -\overline{W})d\tau,\\[2.0mm]
d\overline{W}&=& (b\overline{V} -c \overline{W})d\tau.
\end{array}\right.
\end{split}
\end{equation}
To analyze the bifurcation characteristics at the fixed point $V_e$, one examines the eigenvalues of the corresponding Jacobian matrix, which is defined by
\begin{equation}\label{eq:13}
J=\left[  \begin{array}{cc} \frac{1}{\varepsilon}\big(-3V_e^2 +2(a+1)V_e-a \big)&\:\:\:\:\: -\frac{1}{\varepsilon}\\
b&\:\:\:\:\: -c \end{array} \right].
\end{equation}
The stability of the fixed points $V_e$ is determined by the signs of the trace and determinant of $ J $. For a fixed point $ V_e $ to be stable, it suffices that $\operatorname{tr}J < 0 $ and $ \det J > 0 $. Given that $ \varepsilon $, $ c >0$, we achieve $\operatorname{tr}J < 0 $ and $ \det J > 0 $  only if
\begin{equation}\label{eq:14} 
-3V_e^2 +2(a+1)V_e  -a<0.
\end{equation} 

Equation \eqref{eq:14} indicates that the fixed point $ V_e $ must lie on the decreasing segments of the $v$-nullcline of Eq. \eqref{eq:4}, {\it i.e.,} the curve of solutions to the cubic equation  $w = -v^3 + (a+1)v^2 - av$, where $dw/dv<0$, in order to be stable. As we shall see later, the inequality in Eq.\eqref{eq:14} is a sufficient, but not necessary condition for a fixed point $V_e$ to be stable.

When there are three distinct fixed points $V_e = \{V_1, V_2, V_3\}$ as given in Eq.\eqref{eq:7}, this condition applies to the leftmost and rightmost fixed points, which would both be stable, while the middle fixed point, {\it i.e.,} $V_2$ would be unstable. 
\begin{proposition}\label{proposition 2}
If $ a > 0 $, then the fixed point $(V_1,W_1)=(0,0)$ of Eq. \eqref{eq:4} is stable. 
\end{proposition}
It is worth noting that Proposition \ref{proposition 2} is a sufficient (but not necessary) condition for $(V_1,W_1)=(0,0)$ to be stable.

To obtain complex conjugate eigenvalues of the Jacobian in Eq.\eqref{eq:13} with vanishing real part, we must have that  $(\operatorname{tr}J)^2-4\det J<0$. Hence,  the conditions for complex conjugate eigenvalues are given by
\begin{equation}\label{eq:15}
\begin{split}
\left\{\begin{array}{lcl}
\frac{1}{\varepsilon}\big(-3V_e^2 +2(a+1)V_e -a\big) +c
-2\sqrt{b/\varepsilon}<0,\\[2.0mm]
\frac{1}{\varepsilon}\big(-3V_e^2 +2(a+1)V_e -a\big) +c
+2\sqrt{b/\varepsilon}>0,
\end{array}\right.
\end{split}
\end{equation}
and the condition for a vanishing real part
\begin{equation}\label{eq:16} 
\frac{1}{\varepsilon}\big(-3V_e^2
+2(a+1)V_e  -a\big)-c=0.
\end{equation} 
\begin{proposition}\label{proposition 3}
A fixed point $(V_e,W_e)$ of Eq. \eqref{eq:4} exhibits an Andronov-Hopf bifurcation if and only if the conditions in Eqs. \eqref{eq:15} and \eqref{eq:16} are simultaneously satisfied.
\end{proposition}
To satisfy Eqs.\eqref{eq:15} and \eqref{eq:16} and hence have a 
Andronov-Hopf bifurcation, the parameters $b,c\ (>0)$ need to satisfy condition
\begin{equation}\label{eq:17} 
c^2 <b/\varepsilon,
\end{equation} 
which is satisfied when $\varepsilon>0$ is sufficiently small, since we have  $b>0$.

In the subsequent analysis, we assume the hypothesis of Proposition \ref{proposition 1}. At the fixed point $V_1=0$, the trace and determinant of the Jacobian in Eq.\eqref{eq:13} are  
\begin{equation}\label{eq:18}
\begin{split}
\left\{\begin{array}{lcl}
\operatorname{tr} J&=&-c-a/\varepsilon,\\[2.0mm]
\det J&=&(ac+b)/\varepsilon.
\end{array}\right.
\end{split}
\end{equation}

\begin{proposition}\label{proposition 4}
If $b>0$, $ c>0$, $a>-b/c$, and  $a>-\varepsilon c$, then the fixed point $(V_1,W_1)=(0,0)$ of Eq. \eqref{eq:4} is stable and in the limit $\varepsilon \to 0$, this stability persists only for $a\ge0$.
\end{proposition}

Recall that Proposition \ref{proposition 2} gave a sufficient ($a>0$), but not necessary condition for the fixed point $(V_1,W_1)=(0,0)$ to be stable. For example, Proposition \ref{proposition 4} gives a different set of sufficient conditions for the stability of $(V_1,W_1)=(0,0)$.

The $v$-nullcline, given by the cubic equation $W = -V^3 +(a+1)V^2 - aV$ loses normal hyperbolicity at its local maximum $V_+$ and minimum $V_-$ points, each located at 
\begin{equation}\label{eq:19}
V_\pm=(a+1)/3 \pm \sqrt{(a+1)^2/9 - a/3}. 
\end{equation}
It is worth noting that $V_- < V_1=0$ if and only if $a < 0$.
\begin{proposition}\label{proposition 6}
The fixed point $(V_1,W_1)=(0,0)$ of Eq. \eqref{eq:4} exhibits two Andronov-Hopf bifurcations at $V_{h}^-$ and $V_{h}^{+}$, given respectively, as per Eq.\eqref{eq:16}, by
\begin{equation}\label{eq:20}
V_\text{h}^{\pm}=(a+1)/3\pm\sqrt{(a+1)^2/9-(a+\varepsilon c)/3},
\end{equation}
if and only if $c^2<b/\varepsilon$ and $3\varepsilon c \le a^2-a +1$.
\end{proposition}
Given that $\varepsilon c>0$, the left Hopf bifurcation $V_\text{h}^{-}$ in Eq. \eqref{eq:20} is to the right of the local minimum $V_-$ in Eq. \eqref{eq:19}, {\it i.e.,} $V_-<V_\text{h}^{-}$, and hence $V_\text{h}^{-}$ is on middle (unstable) increasing segment of the $v$-nullcline  $W = -V^3 +(a+1)V^2 - aV$. 

Whenever the fixed point $V_1=0$ is located on the left decreasing segment of the $v$-nullcline, specifically to the left of its local minimum $V_-$ ({\it i.e.,} $0=V_1<V_-$), it remains stable, as per Eq. \eqref{eq:14}. However, as $\varepsilon\to0$ in Eq. \eqref{eq:20}, the left Hopf bifurcation point $V_\text{h}^{-}$, {\it i.e.,} where the fixed point $V_1=0$ loses its stability, shifts toward the local minimum  $V_-$ in Eq. \eqref{eq:19}. 

For $\varepsilon > 0$, the left Hopf bifurcation $V_\text{h}^{-}$ occurs to the right of the local minimum $V_-$ (\textit{i.e.,} $V_-<V_\text{h}^{-}$). Hence, the stability of the fixed point $V_1=0$ extends slightly into the increasing segment of the $v$-nullcline.  We notice that for the fixed point $V_e=V_1=0$, Eq. \eqref{eq:16} has no solution under the conditions $\varepsilon>0$ and $ c > 0 $, when $ a > 0 $. But when $ a < 0 $, Eq. \eqref{eq:16} is satisfied for $V_e=V_1 = 0 $ if
\begin{equation}\label{eq:21}
\varepsilon_\text{h} = - a/c.
\end{equation}
Hence, to position the stable and unique fixed point $V_1 = 0$ to the right of the local minimum $V_-$ ({\it i.e.,} $V_-<V_1=0$), it suffices to choose specific values of $a < 0 $ and $c > 0 $ such that $a $ and $c $ satisfy the assumptions of Proposition \ref{proposition 1} and $0 = V_1 < -a/c = \varepsilon_\text{h}$.  This shows that when $a < 0 $ and as $\varepsilon>0$ increases, the fixed point $V_1 = 0 $ is positioned to the right of the local minimum $V_-$ ({\it i.e.,} $V_- < V_1 = 0$), and also in the region ({\it i.e.,} $V_- <V_1<V_\text{h}^{-}$) where it eventually loses its stability, via a Hopf bifurcation when $\varepsilon=\varepsilon_\text{h}$. 

This subcritical Hopf bifurcation (occurring as $\varepsilon$ increases) produces:
\textit{(i)}  a stable limit cycle  $\Gamma_\LC$ around the stable fixed point \(X_\FP=(V_1,W_1)\) and \textit{(ii)} an unstable limit cycle (saddle) $\Gamma_\SP$ between \(X_\FP\) and \(\Gamma_\LC\) emerging by topological necessity, with the unstable cycle confirming the bifurcation's subcritical nature. Combining the assumptions of Propositions \ref{proposition 1}, \ref{proposition 6}, and the inequality \(V_- < V_1 < V_{\mathrm{h}}^{-}\), we establish the following co-existence result for the FHN neuron model in Eq. \eqref{eq:4}:

\begin{proposition}\label{proposition 7}
The deterministic ($\sigma=0$) FHN neuron model Eq. \eqref{eq:4} with \(b > 0\), \(c > 0\), and \(\varepsilon > 0\) admits bistability between a stable unique fixed point
\(X_\FP=(V_1,W_1)=(0,0)\) and a stable \(T_\LC\)-periodic limit cycle $\Gamma_\LC$, provided the following conditions hold simultaneously:
\begin{enumerate}[label=(\roman*)]
    \item  $(a-1)^2/4 < b/c$,
    \item \(a > -b/c\) and \(a > -\varepsilon c\),
    \item \(V_- < V_1 < V_h^{-}\),
\end{enumerate}
where \(V_-\) and $V_h^{-}$ denote, respectively, the local minimum of the \(v\)-nullcline and the \(V\)-coordinate of the left Hopf bifurcation point, as given in Eqs.~\eqref{eq:19} and \eqref{eq:20}, and where the corresponding subcritical Hopf bifurcation threshold in the parameter \(\varepsilon\) is given by Eq. \eqref{eq:21}.
\end{proposition}

In this paper, we denote the bistability regime of Proposition \ref{proposition 7} by the wedge $\Omega_{\text{bi}}$:
\begin{equation}
\Omega_{\text{bi}} := \left\{(a,\varepsilon)\in\mathbb R^2 : 
\varepsilon_h(a)<\varepsilon<\varepsilon_{sn}(a) \quad
\text{or}\quad a_h(\varepsilon)<a<a_{sn}(\varepsilon)
\right\},  
\end{equation}
where $\varepsilon_h(a)=-a/c$ is the subcritical Hopf curve, $a_h(\varepsilon)=-c\varepsilon$ is the vertical Hopf condition, $\varepsilon_{sn}(a)$ and  $a_{sn}(\varepsilon)$ are the saddle-node–of–limit-cycles curves.
 For all $(a,\varepsilon)\in\Omega_{\text{bi}}$, Eq. \eqref{eq:2} (with $\sigma=0$) possesses exactly three invariant sets: $X_\FP=(V_1,W_1)=(0,0)$, $\Gamma_\LC$, and  $\Gamma_\SP$.

 \begin{figure}[h]
\centering
\includegraphics[width=7.75cm,height=4.5cm]{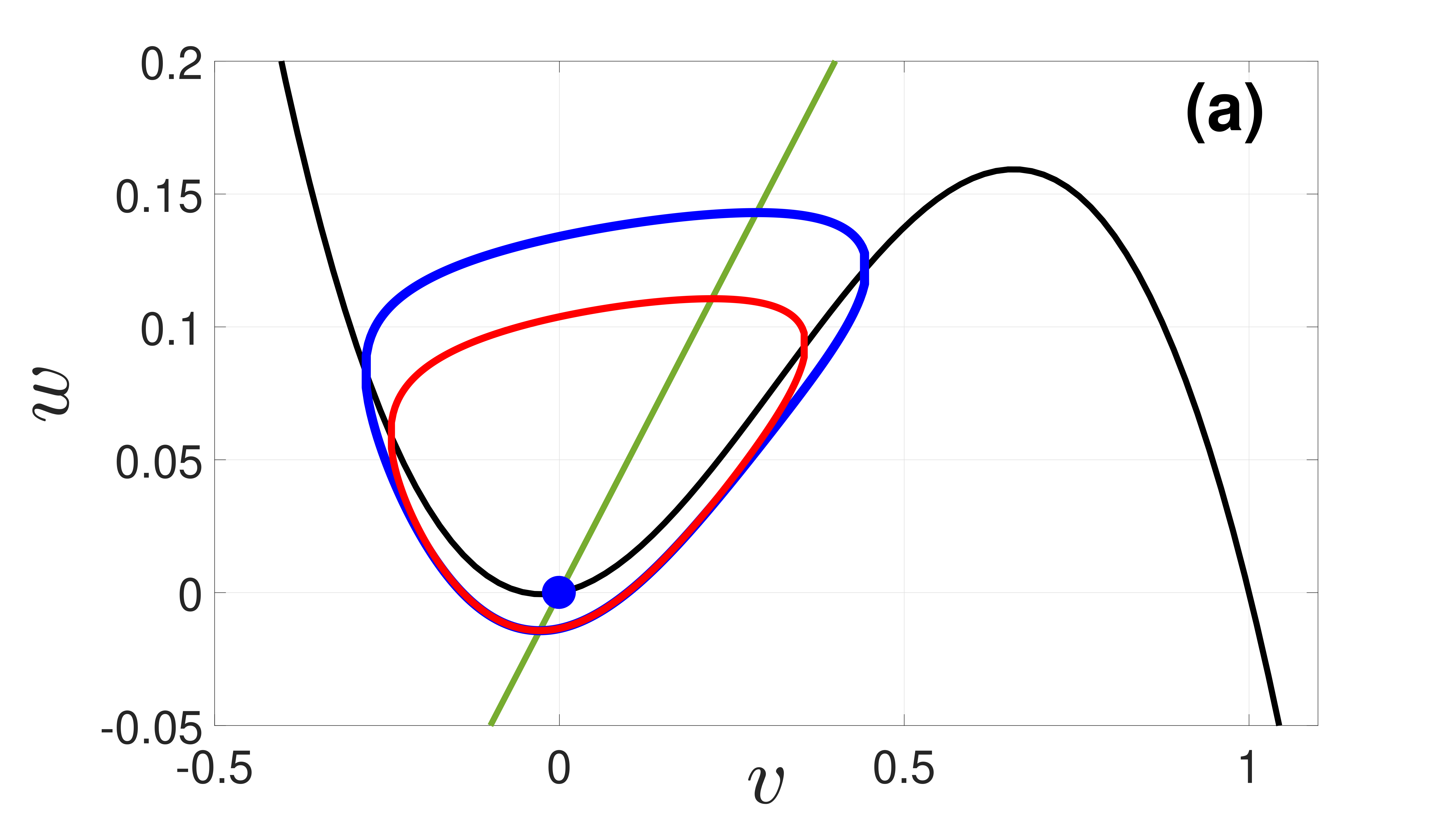}\includegraphics[width=7.5cm,height=4.5cm]{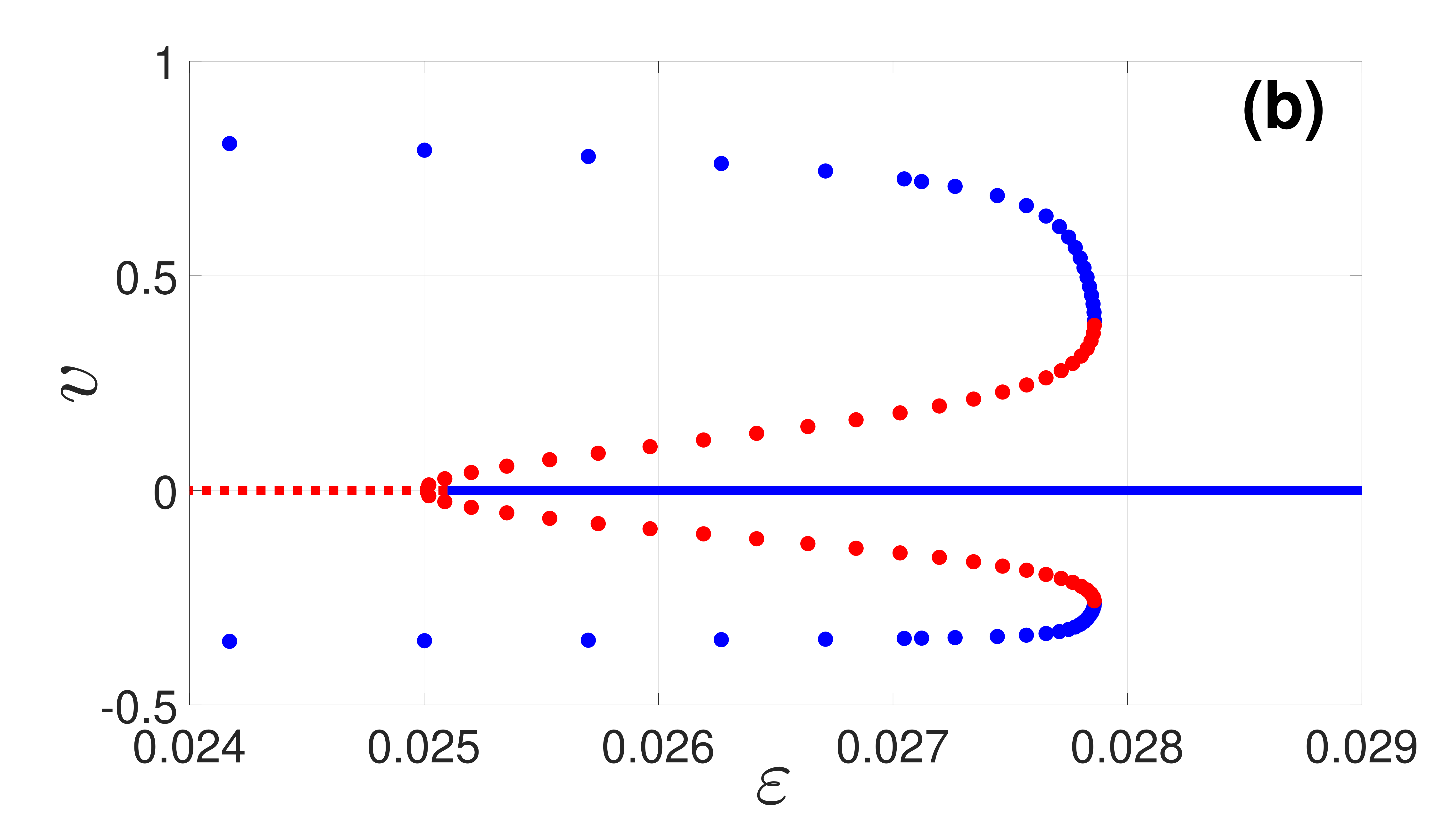}
\includegraphics[width=7.75cm,height=4.5cm]{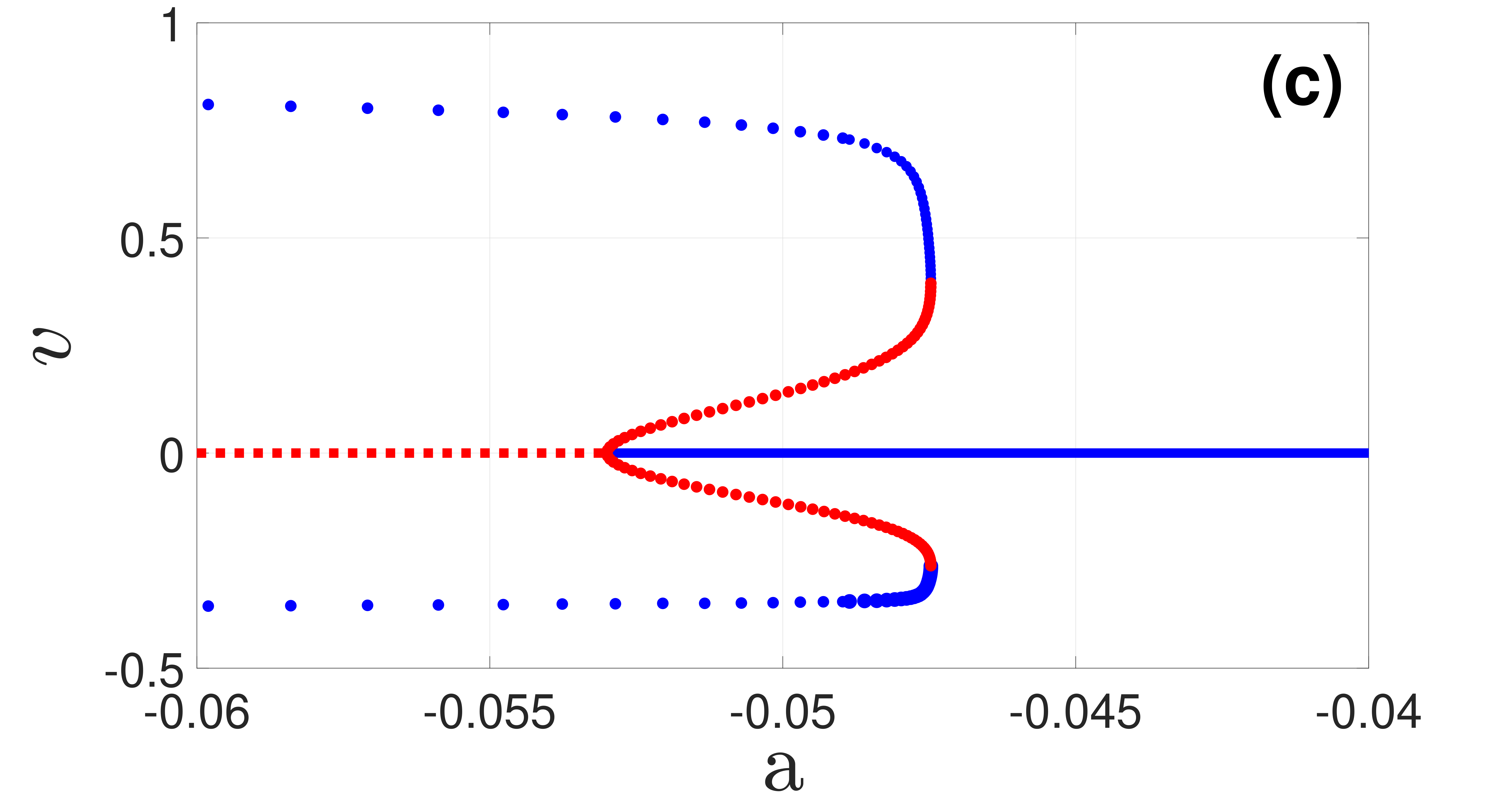}\includegraphics[width=7.75cm,height=4.5cm]{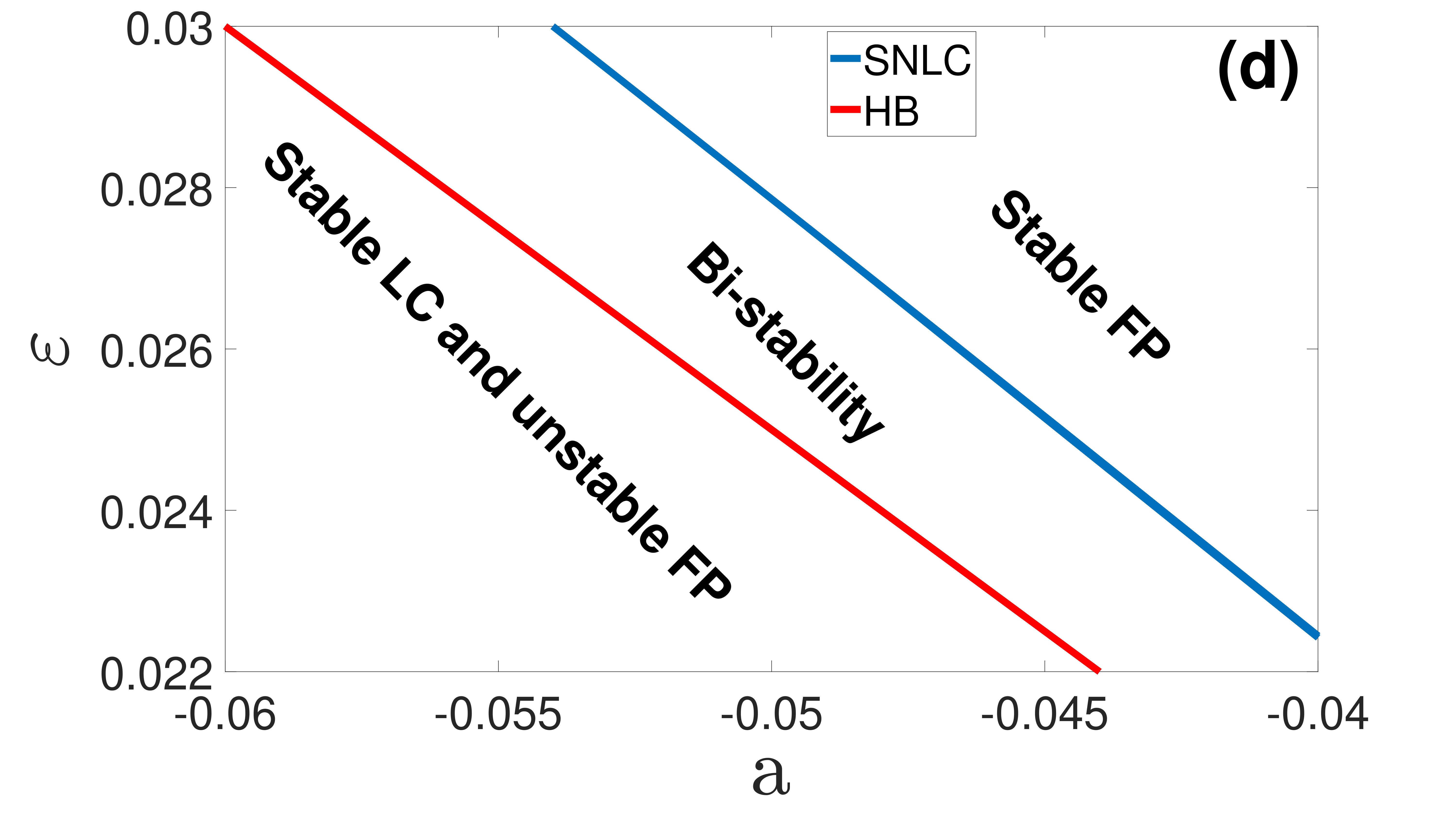}
\caption{(a) Phase portrait: $v$-nullcline (black) and $w$-nullcline (green) intersect at the fixed point (0,0) (blue dot). Unstable (red) and stable (blue) limit cycles encircle (0,0). (b) - (c) Bifurcation diagrams: $v$ vs. $\varepsilon$ ($a=-0.05$) and $v$ vs. $a$ ($\varepsilon=0.0265$) showing bistability (stable FP [blue line], stable LC [blue dots], unstable LC [red dots]). (d) $a$-$\varepsilon$ bifurcation: LC-FP bistability region bounded by subcritical HB and SNLC curves. Parameters: $b=1.0$, $c=2.0$.}
\label{fig:1}
\end{figure}

We also denote the basins of attraction of the stable fixed point  and limit cycle by:
\begin{equation}\label{basin}
\begin{split}
\left\{\begin{array}{lcl}
\mathcal{B}_\FP := \{Y_0 \in \mathbb{R}^2 :\limsup\limits_{\tau\rightarrow\infty}\|Y_{\tau} - X_\FP \| =0\},\\[2.0mm]
\mathcal{B}_{\LC} := \{Y_0\in \mathbb{R}^2 : \limsup\limits_{\tau\rightarrow\infty}\mbox{\rm dist}( Y_{\tau},\Gamma_\LC)=0 \},
\end{array}\right.
\end{split}
\end{equation}
respectively, where $Y_{\tau}=(V_{\tau},W_{\tau})\in\mathbb{R}^2$ is a trajectory of Eq. \eqref{eq:4} with initial condition $Y_0=(V_0,W_0)\in\mathbb{R}^2$. Furthermore, the separatrix is the unstable limit cycle $\Gamma_{\SP}=\partial\mathcal{B}_\FP$.

Furthermore, we denote by
\[
\chi_{\LC}(\tau)
   :=\mathbbm{1}_{\{Y(\tau)\in\mathcal B_{\LC}\}}
   =\begin{cases}
      1,& Y(\tau)\in\mathcal B_{\LC},\\[4pt]
      0,& Y(\tau)\notin\mathcal B_{\LC},
     \end{cases}
\]
the characteristic (indicator) function of the stable limit–cycle basin
\(\mathcal B_{\LC}\).

Using Proposition \ref{proposition 7}, we can precisely determine the parameter intervals of $a$ and $\varepsilon$ in which the stable fixed point $X_\FP$ and stable limit cycle $\Gamma_\LC$ co-exist. Proposition \ref{proposition 7} gives results shown in Figs. \ref{fig:1}(a)-(d). Figure~\ref{fig:1}(a) shows a phase portrait where a stable fixed point
coexists with a stable limit cycle.
The one-parameter bifurcation curves in
Figs.~\ref{fig:1}(b) and (c) and the two-parameter map in
Fig.~\ref{fig:1}(d) bound this bi-stable window, delimited by a
subcritical Hopf (HB) and a saddle-node of limit cycles (SNLC).  For \(a=-0.05\) the thresholds are \(\varepsilon_{\mathrm{h}}=0.025\) and
\(\varepsilon_{\mathrm{sn}}=0.027865\); for \(\varepsilon=0.0265\) they are
\(a_{\mathrm{h}}=-0.053\) and \(a_{\mathrm{sn}}=-0.0474726\).

\section{Monte Carlo simulations: Effects of $a$, $\varepsilon$, and initial conditions on ISR}
\label{sec: stochastic analysis}
We study how the excitability parameter \( a \) and timescale separation \( \varepsilon \) influence the non-monotonic characteristics of ISR within their bistability ranges. For system Eq. \eqref{eq:2}, we quantify ISR via the mean firing rate (Hz):
\begin{equation}\label{eq:3m1}
    \langle r \rangle = \lim_{T \to \infty} \frac{n_{\text{spike}}(T)}{T},
\end{equation}
where \( n_{\text{spike}}(T) \) counts spikes (upcrossings of \( v(t) \) through \( v_{\text{th}} = 0.25 \)) in \([0,T]\).  Furthermore, assuming each full traversal of the limit cycle produces one spike, the mean firing rate can therefore be approximated (where the approximation uses the weak dependence of $T_\LC$ on $\sigma$ at small noise) by:
\begin{equation}\label{eq:mfr_proxy}
\langle r\rangle\;\approx\;\frac{\mu_{\sigma}(\mathcal \mathcal{B}_{\LC})}{T_\LC},  
\end{equation}
where \(T_\LC=T_\LC(a,\varepsilon)\) is the deterministic period of $\Gamma_\LC$ and $\mu_{\sigma}(\mathcal B_{\LC})$ is the occupation probability of the 
limit-cycle basin $\mathcal B_{\LC}$.

The occupation probability $\mu_{\sigma}(\mathcal B_{\LC})$ is controlled by the residence times in the two deterministic basins. We define the first-exit times
\begin{align}
\tau_\FP &= \inf\{t>0:\,X_t\notin\mathcal{B}_\FP \,|\, X_0=X_\FP\}, \\
\tau_\LC &= \langle\tau_\LC^{x}\rangle_{x\in\Gamma_\LC}\quad, \quad
\tau_\LC^{x} = \inf\{t>0:\,X_t\notin\mathcal{B}_\LC \,|\, X_0=x\in\Gamma_\LC\},
\end{align}
where \(X_t=(v_t,w_t)\) is the stochastic trajectory and
\(\langle\cdot\rangle_{x\in\Gamma_\LC}\) denotes averaging over one deterministic period of the stable limit cycle. The standard residence-time interpretation of ISR is that, at an intermediate noise level, trajectories spend comparatively longer times in $\mathcal{B}_{\FP}$ than in $\mathcal{B}_{\LC}$, thereby suppressing spiking \cite{tuckwellAnalysisInverseStochastic2012,gutkinInhibitionRhythmicNeural2009,tuckwellInhibitionModulationRhythmic2009,li2026role}.

Figures~\ref{fig:4}(a1) and~\ref{fig:4}(a2) show that, for the finite
simulation time used here and for initial condition
$X_0\in\mathcal B_{\LC}$, the occupation probability
$\mu_{\sigma}(\mathcal B_{\LC})$ exhibits a non-monotone dependence on
$\sigma$ throughout the bistable intervals of the wedge $\Omega_{\text{bi}}$. The corresponding dips become
deeper as $(a,\varepsilon)$ approaches the SNLC boundary
$(a_{\mathrm{sn}},\varepsilon_{\mathrm{sn}})$ and become shallower near the
Hopf boundary $(a_{\mathrm{h}},\varepsilon_{\mathrm{h}})$. If these curves were
considered in isolation, one might conclude that ISR occurs throughout the
bistable wedge whenever the trajectory is initialized in the basin of the
stable limit cycle, and that its strength is controlled mainly by proximity to
the SNLC boundary.

Figures~\ref{fig:4}(b1) and~\ref{fig:4}(b2), computed over the same simulation
horizon but with $X_0\in\mathcal B_{\FP}$, show a different finite-time
picture. In this case, non-monotone occupation curves are visible only near the
Hopf boundary, whereas closer to the SNLC boundary the curves appear monotone
over the simulated time window. Taken together,
Figs.~\ref{fig:4}(a1)--\ref{fig:4}(b2) therefore suggest an apparent dependence
of ISR on the initial basin: ISR-like dips seem to occur broadly for
$X_0\in\mathcal B_{\LC}$, but only near the Hopf boundary for
$X_0\in\mathcal B_{\FP}$.

This apparent dependence is a finite-time effect. In a bistable stochastic
system, the long-time occupation of $\mathcal B_{\LC}$ is determined by the
balance between the two noise-induced escape mechanisms: escape from the
fixed-point basin and escape from the limit-cycle basin. Over a finite
observation window, one of these transitions may be so rare that it is not
adequately sampled. The measured occupation probability may then retain a
visible dependence on the initial basin and may display an apparent dip that
does not persist in the invariant occupation statistics.

Consequently, finite-time non-monotonicity alone is not sufficient to identify
genuine ISR. We regard an ISR minimum as asymptotically meaningful only when
the occupation curves obtained from initial conditions in
$\mathcal B_{\FP}$ and $\mathcal B_{\LC}$ coincide, or nearly coincide, and
both exhibit a minimum at an intermediate noise amplitude. In the next section,
we prove that the stochastic FHN system admits a unique invariant probability
measure. It follows that, for fixed $(a,\varepsilon)$ and $\sigma>0$, the
long-time occupation probability $\mu_\sigma(\mathcal B_{\LC})$ is independent
of the initial condition. Thus, the asymptotic occurrence of ISR is determined
not by the choice of $X_0$, but by the parameter-dependent balance between the
two escape mechanisms. The quasi-potential analysis and reduced escape-balance theory developed in Sections~\ref{sec:stochastic_analysis} and \ref{sec:_heuristics2} make this balance quantitative and identify which side of
the bistable wedge supports a genuine weak-noise ISR minimum.

\begin{figure}
\centering
\includegraphics[width=7.5cm,height=5.5cm]{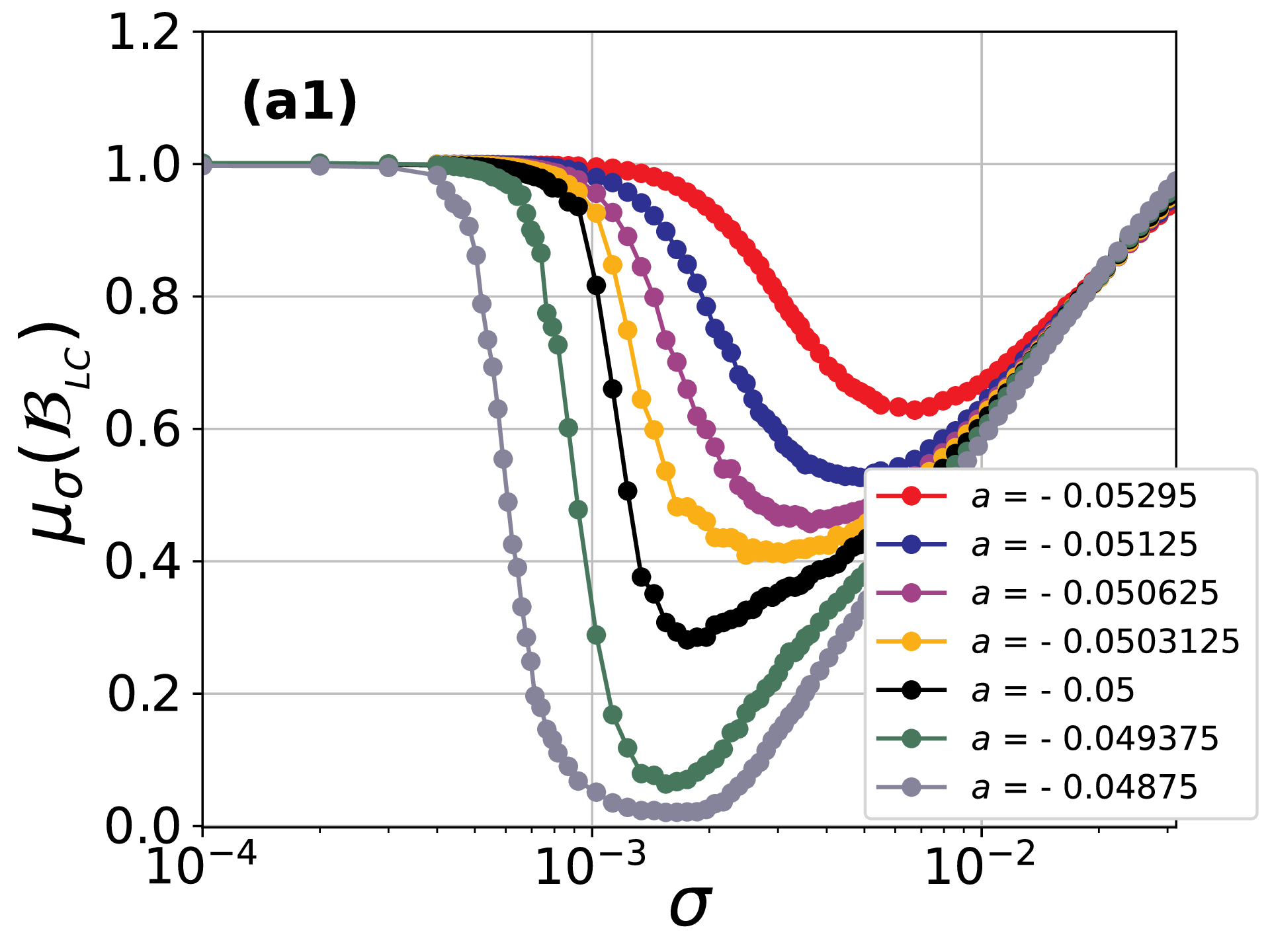}\includegraphics[width=7.5 cm,height=5.5cm]{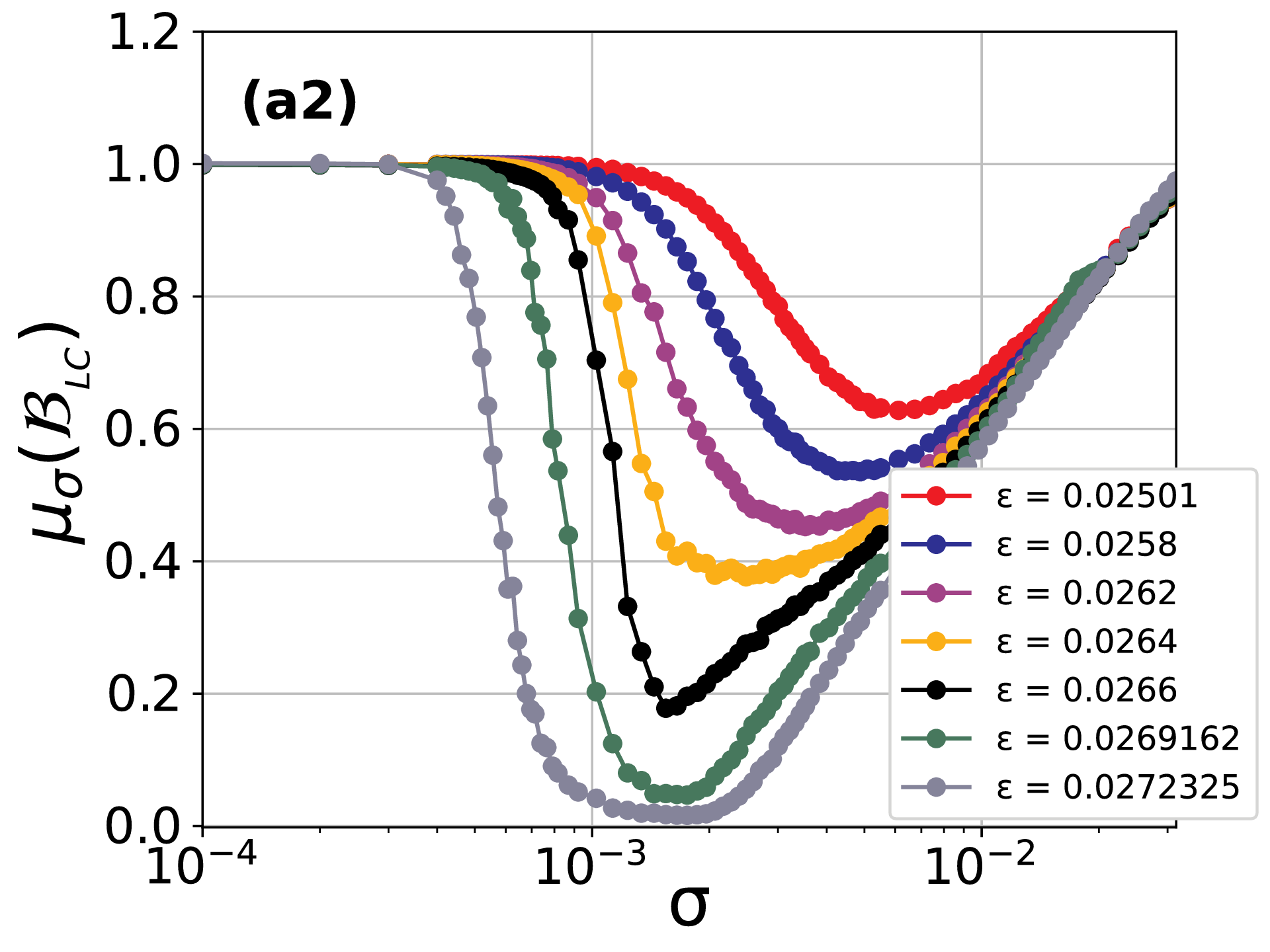}
\includegraphics[width=7.5cm,height=5.5cm]{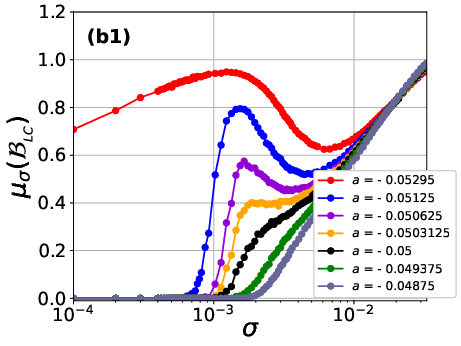}\includegraphics[width=7.5cm,height=5.5cm]{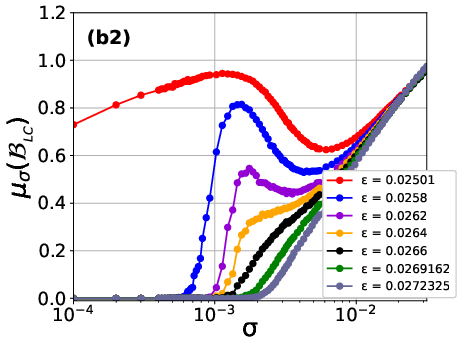}
\caption{Finite-time estimates of the limit-cycle basin occupation probability
$\mu_{\sigma}(\mathcal B_{\LC})$ as a function of the noise amplitude $\sigma$
for different values of $a$ and $\varepsilon$ in the bistable domain $\Omega_{\text{bi}}$. Each curve
is averaged over 300 independent realizations with simulation time
$T=5.0\times10^4$. Panels (a1)--(a2) use
$X_0=(0.75,1.0)\in\mathcal B_{\LC}$, while panels (b1)--(b2) use
$X_0=(0,0)\in\mathcal B_{\FP}$. The comparison shows an apparent finite-time
dependence on the initial basin: non-monotone curves are broadly visible for
$X_0\in\mathcal B_{\LC}$, whereas for $X_0\in\mathcal B_{\FP}$ they appear
mainly near the Hopf boundary. Other parameters are $b=1.0$ and $c=2.0$.}
\label{fig:4}
\end{figure}

\section{Uniqueness of the invariant probability measure and independence of ISR from initial conditions}
\label{invariant_measure}
The evolution of the  probability density \( P_t(x,y) \) of the solution $X_t=(v_t,w_t)$ to the SDE Eq. \eqref{eq:2}, with coefficient functions given by Eq. \eqref{eq:3}, is governed by the 
Fokker--Planck equation
\begin{equation}\label{eq:FP_equation}
\frac{\partial P_t}{\partial t} = \mathcal{L}^* P_t,
\end{equation}
with initial condition
\begin{equation}\label{eq:initial_condition}
P_0(x,y) = \delta_{y}(x), 
\end{equation}
where \( y=(v_0, w_0) \) is the initial state, and natural boundary conditions
\begin{equation}\label{eq:boundary_conditions}
\lim_{|x|\to\infty} P_t(x,y)=0, \qquad \lim_{|x|\to\infty} \|\nabla_x P_t(x,y)\|=0.
\end{equation}
Here, the 
forward Kolmogorov--Fokker--Planck (KFP) operator $\mathcal{L}^*$ (adjoint of the infinitesimal generator $\mathcal L$) is defined as\begin{equation}\label{eq:KFP_operator}
\mathcal{L}^* P_t = \frac{\sigma^2}{2} \frac{\partial^2 P_t}{\partial v^2} - \frac{\partial}{\partial v}\bigl[(v(a - v)(v - 1) - w)P_t\bigr] - \frac{\partial}{\partial w}\bigl[\varepsilon(b v - c w)P_t\bigr],
\end{equation}
for \( P_t \in L^2(\mathbb{R}^2) \).

\begin{theorem}[Invariant measure of the FHN SDE]
\label{thm:FHN-ergodic}
Fix parameters $a\in\mathbb R$ and $b,c,\varepsilon,\sigma>0$.
Let $X_t=(v_t,w_t)$ be the unique strong solution of Eq. \eqref{eq:2} with
initial state $X_0=(v_0,w_0)\in\mathbb R^{2}$, and denote by
$P_t\varphi(x)=\mathbb E_x[\varphi(X_t)]$ the associated Markov semigroup.
Let $\mathcal B_{\LC}\subset\mathbb R^{2}$ be the attraction basin of the limit cycle 
defined in Eq.~\eqref{basin}, and set
$\chi_{\LC}(t)  :=\mathbbm{1}_{\{X_t\in\mathcal B_{\LC}\}}$.
Then
\begin{enumerate}[label=\textup{(\roman*)}]
\item (\textbf{Global well-posedness})  
      The SDE Eq.~\eqref{eq:2} admits a unique and non-explosive strong solution for every initial $x\in\mathbb R^{2}$.  
\item (\textbf{Existence/Uniqueness of invariant measure})  
      There exists a unique invariant probability measure $\mu_\sigma$ on $\mathbb R^{2}$
      that is absolutely continuous with respect to the Lebesgue measure.
Its density $P_{\mbox{\rm\tiny stat}}(x)$, called the stationary density, is smooth, locally strictly positive, and satisfies in the distributional sense 
\begin{equation}\label{eq:stationary_density}
\mathcal{L}^* P_{\mbox{\rm\tiny stat}} = 0.
\end{equation}
\item (\textbf{Exponential rate of convergence})
      There exist constants $C,\rho>0$ and a Lyapunov function $V\ge1$ such that for the weighted norm
      \[
         \|P_t(x,y)dx-\mu_\sigma(dx)\|_{V}
         \le C\,V(y)\,e^{-\rho t},\qquad y\in\mathbb R^{2},\;t\ge0,
      \]
      where $\|m\|_{V}:=\sup_{|\varphi|\le V}|m(\varphi)|$ for a signed measure $m$ on $\mathbb{R}^2$.
\item (\textbf{Ergodic law of large numbers})  
      For every $\varphi\in L^{1}(\mu_\sigma)$ and initial condition $x_0\in\mathbb{R}^{2}$,
      \[
         \frac1T\int_0^T \varphi(X_t)\,dt
         \xrightarrow[T\to\infty]{\mathrm{a.s.}}
         \int_{\mathbb R^{2}}\varphi(x)\,\mu_\sigma(dx).
      \]
      In particular, noting that $\mu_\sigma(\partial\mathcal B_{\LC})=0$, we have
            \begin{equation}\label{birkoff}
             \frac{1}{T}\int_{0}^{T}\chi_{\LC}(t)\,dt
         \xrightarrow[T\to\infty]{\mathrm{a.s.}}
         \int_{\mathbb R^{2}}\mathbbm{1}_{\{x\in\mathcal B_{\LC}\}}\,\mu_\sigma(dx)
         =\mu_\sigma(\mathcal B_{\LC}).  
      \end{equation}
\end{enumerate}
\end{theorem}
The proof of Theorem \ref{thm:FHN-ergodic} is presented in the Appendix \ref{appendix}. Thus, for fixed $(a,\varepsilon)\in\Omega_{\mathrm{bi}}$ and $\sigma>0$, the spiking dynamics of Eq.~\eqref{eq:2} with arbitrary initial data admits a unique invariant measure $\mu_{\sigma}$. The long-time $\sigma$–dependence of $\mu_\sigma(\mathcal B_{\LC})$ is therefore determined solely by $(a,\varepsilon)$, and not by the choice of initial basin. 
Consequently, as $T\to\infty$, Monte Carlo estimates obtained from initial
conditions in $\mathcal B_{\FP}$ and $\mathcal B_{\LC}$ converge to the same
invariant occupation probability $\mu_\sigma(\mathcal B_{\LC})$. Hence, the
initial condition plays no role in the asymptotic manifestation of ISR; any
apparent initial-basin dependence in Fig.~\ref{fig:4} is a finite-time effect. Figure~\ref{fig:5} illustrates this convergence for a parameter point in the genuine ISR regime (\textit{e.g.,} the red ISR curves in Figs. \ref{fig:4}(a1) and (b1) with a minimum at an intermediate noise $\sigma$).

\begin{figure}
\centering
\includegraphics[width=5.0cm,height=4.0cm]{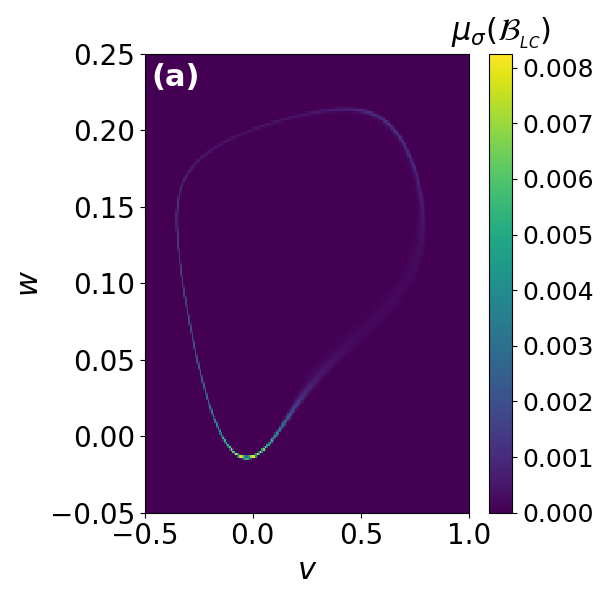}\includegraphics[width=5.0 cm,height=4.0cm]{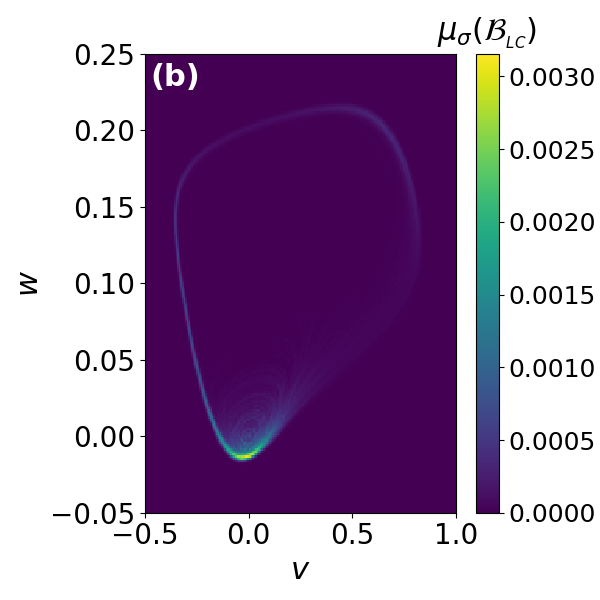}\includegraphics[width=5.0cm,height=4.0cm]{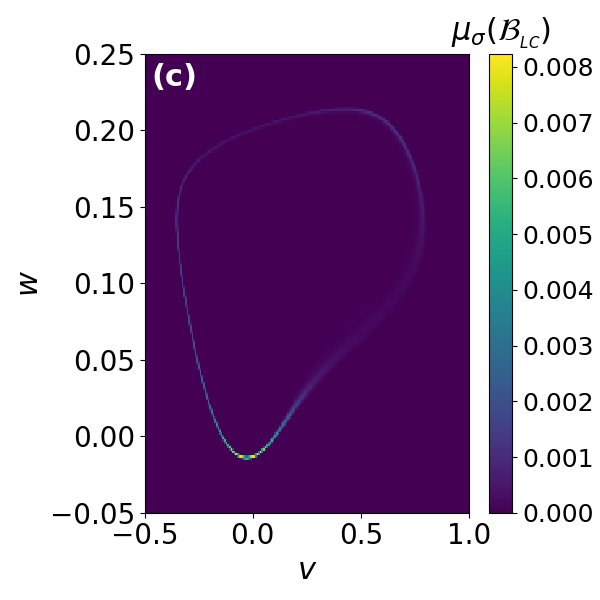}
\includegraphics[width=5.0cm,height=4.0cm]{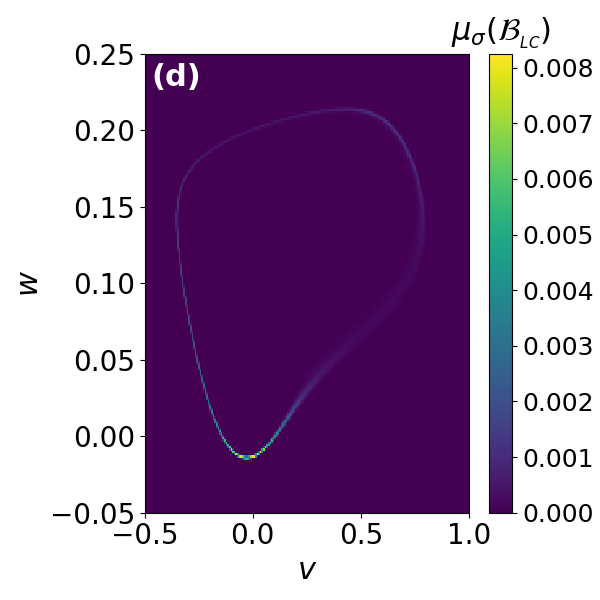}\includegraphics[width=5.0cm,height=4.0cm]{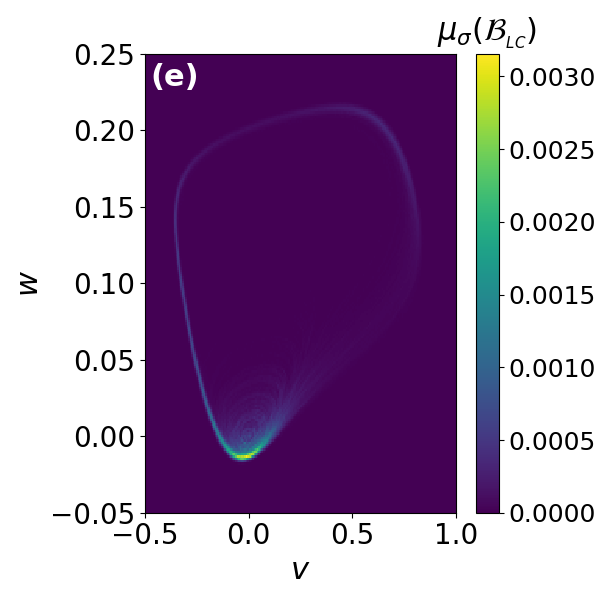}\includegraphics[width=5.0cm,height=4.0cm]{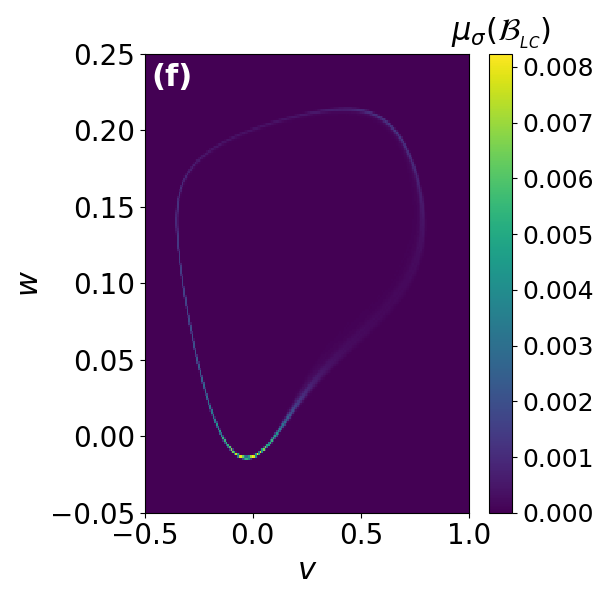}
\caption{
Probability density in phase space for different noise amplitudes.
(a)--(c): phase-space probability densities obtained from trajectories initialized in
\(\mathcal B_{\LC}\) for weak \((\sigma=0.001)\), intermediate \((\sigma=0.006)\),
and strong \((\sigma=0.037)\) noise intensities.
(d)--(f): corresponding phase-space probability densities obtained from trajectories
initialized in \(\mathcal B_{\FP}\).
The associated occupation probability \(\mu_\sigma(\mathcal B_{\LC})\) is obtained by
integrating the density over the limit-cycle basin \(\mathcal B_{\LC}\).
The redistribution of probability mass between \(\mathcal B_{\FP}\) and
\(\mathcal B_{\LC}\) illustrates the non-monotonic behavior characteristic of ISR.
Parameters: \(a=-0.05295\), \(\varepsilon=0.0265\),
\(\Delta S(a,\varepsilon)=7.9107\times10^{-6}\),
\(b=1.0\), and \(c=2.0\).
}
\label{fig:5}
\end{figure}

\section{Quasi‑potentials of the attractors as function of $a$ and $\varepsilon$}
\label{sec:stochastic_analysis}
In this section, we compute the quasi-potential barriers associated with escape
from the stable fixed point and the stable limit cycle. These barriers determine
the leading exponential contribution to the mean first-exit times and provide
the parameter-dependent input for the reduced escape-balance theory developed
in Section~\ref{sec:_heuristics2}. Thus, the influence of the excitability
parameter $a$ and the timescale parameter $\varepsilon$ on ISR is encoded, to
leading order in the weak-noise regime, in the quasi-potential costs required
to reach the separatrix $\Gamma_{\SP}$ from the two attracting sets. This
provides a quantitative mechanism for how intrinsic neuronal parameters shape
noise-induced switching between quiescent and spiking dynamics. Elucidating this mechanism is of particular significance for the broader neuroscience community.

 We recall that, for small-noise stochastic differential equations, mean first-exit times from metastable basins are governed, to leading exponential order, by the corresponding quasi-potential barriers \cite{mil1995first,lin2019quasi}. In the present problem, these are the escape
costs from the fixed-point basin $\mathcal B_{\FP}$ and from the limit-cycle
basin $\mathcal B_{\LC}$ to the separatrix $\Gamma_{\SP}$. 

In Eq.~\eqref{eq:2},
the noise acts only on the fast variable $v$, so the diffusion matrix of the
SDE is given by
$D=\sigma^2 NN^\top,$ where $N=(1,0)^\top$.  For the large-deviation analysis, we introduce the $\sigma$-independent matrix
$G:=NN^\top.$  In the Freidlin--Wentzell theory, the action functional is defined as \cite{kuehn2015multiple,freidlin2012random,berglund2006noise}
\begin{equation}\label{action_1}
A_T(\phi)
=\frac{1}{2}\int_0^T \bigl\| \dot{\phi}(t)-F(\phi(t)) \bigr\|^2_{G(\phi(t))^{-1}}\,dt,
\end{equation}
where $\phi:[0,T]\to\mathbb{R}^2$ is an absolutely continuous path, and $\|v\|_{G^{-1}}^{2} := v^{\mathsf T} G^{-1} v$. 

In the rank-one case, as in our system, the action functional reduces to \cite{lin2019quasi}
\begin{equation}\label{action}
A_T(\phi)
=\inf_{u:\,u=\dot{\phi}-F(\phi)} \frac{1}{2}\int_0^T \|u(t)\|^2\,dt
=\frac{1}{2}\int_0^T \bigl(\dot{\phi}_v-f(\phi)\bigr)^2\,dt,
\end{equation}
subject to the feasibility constraint $\dot{\phi}_w = g(\phi)$; otherwise $A_T(\phi) = +\infty$.

The quasi-potential quantifies the minimal cost required to go against the deterministic flow from an attractor (either a fixed point or a stable limit cycle) to a point $x \in \mathbb{R}^2$. 
Specifically, the quasi-potentials with respect to the fixed point $X_{\FP}$ and the stable limit cycle $\Gamma_\LC$ are defined by
\begin{equation}\label{quasi_pots}
S_\FP(x)
= \inf_{\substack{T>0\\ \phi(0)\in X_{\FP};\,\phi(T)=x}} A_T(\phi),
\qquad
S_\LC(x)
= \inf_{\substack{T>0\\ \phi(0)\in \Gamma_\LC;\,\phi(T)=x}} A_T(\phi).
\end{equation}
Since the deterministic dynamics along the limit cycle have zero action, $S_\LC(x)=0$ for all $x\in\Gamma_\LC$. 
Hence, $S_\LC(x)$ does not depend on the particular choice of the initial point $\phi(0)\in\Gamma_\LC$.

The escape barriers associated with transitions from the stable fixed point 
$X_{\FP}$ and the stable limit cycle $\Gamma_\LC$ 
to the unstable limit cycle $\Gamma_\SP$ 
are determined by the minima of the corresponding quasi-potentials on the separatrix \cite{freidlin2012random,berglund2006noise}:
\begin{equation}\label{eq:barriers-sigma}
S_\FP := \min_{x \in \Gamma_\SP} S_\FP(x),
\qquad
S_\LC := \min_{x \in \Gamma_\SP} S_\LC(x).
\end{equation}

We compute the quasi-potentials \(S_{\FP}=S_\FP(a,\varepsilon)\) and \(S_{\LC}=S_\LC(a,\varepsilon)\) (Eq.~\eqref{eq:barriers-sigma}) numerically using the degenerate-noise geometric Minimum Action Method (gMAM) \cite{lin2019quasi,heymann2008geometric,heymann2008pathways}, applicable here since noise enters only in \(v\) and \(G\) is singular.  

Figures~\ref{fig:6}(a)--(b) show the computed quasi-potential barriers
$S_{\FP}$ and $S_{\LC}$ as functions of $a$ and $\varepsilon$, each plotted
while fixing the other parameter inside the bistable wedge
$\Omega_{\mathrm{bi}}$. These barriers are interpreted only for
$(a,\varepsilon)\in\Omega_{\mathrm{bi}}$, where the stable fixed point
$X_{\FP}$, the stable limit cycle $\Gamma_{\LC}$, and the separating unstable
limit cycle $\Gamma_{\SP}$ coexist. Thus, all statements about approaching the
Hopf or SNLC boundaries are understood as one-sided limits taken from within
$\Omega_{\mathrm{bi}}$. As the system approaches the Hopf bifurcation from within
$\Omega_{\mathrm{bi}}$
($a \to a_h$ or $\varepsilon \to \varepsilon_h$),
$S_{\FP}$ decreases whereas $S_{\LC}$ increases. Conversely, as the system
approaches the SNLC bifurcation from within $\Omega_{\mathrm{bi}}$
($a \to a_{\mathrm{sn}}$ or $\varepsilon \to \varepsilon_{\mathrm{sn}}$),
$S_{\FP}$ increases while $S_{\LC}$ decreases. Outside
$\Omega_{\mathrm{bi}}$, one of the two attracting sets is absent, and the
quantities $S_{\FP}$, $S_{\LC}$, and
$\Delta S=S_{\LC}-S_{\FP}$ no longer have the interpretation of two competing
metastable escape barriers.

\begin{figure}
\centering
\includegraphics[width=14.0cm,height=6.0cm]{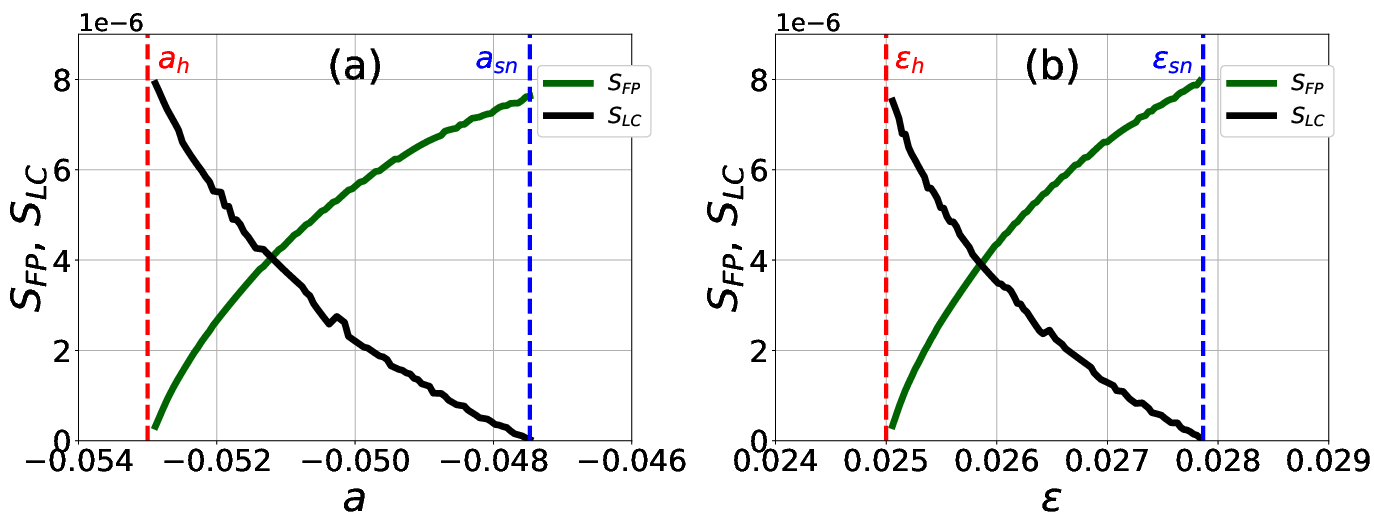}
\caption{Quasi-potentials as a function of parameters $a$ and $\varepsilon$. In (a) \(\varepsilon\!=\!0.0265\) and in (b) $a=-0.05$.  Quasi-potential $S_\FP$ (green) vs. \(a\)/\(\varepsilon\): min at the HB (red dashed lines at \(a_{\mathrm{h}}\!=\!-0.053\), \(\varepsilon_{\mathrm{h}}\!=\!0.025\)) and max at the SNLC (blue dashed lines at  \(a_{\mathrm{sn}}\!=\!-0.0474726\), \(\varepsilon_{\mathrm{sn}}\!=\!0.027865\)). Quasi-potential $S_\LC$ (black) vs. \(a\)/\(\varepsilon\): max at HB and min at SNLC. Other parameters: \(b\!=\!1.0\), \(c\!=\!2.0\).}
\label{fig:6}
\end{figure}

\section{Mechanism and quantitative prediction of the ISR minimum}
\label{sec:_heuristics2}
We now turn to the explanation for the existence of ISR in our FHN neuron model. Importantly, the mechanism leading to ISR is
largely generic: it follows directly from the universal small-noise asymptotics for
mean exit times from metastable attractors.  In our setting, the system possesses two
types of metastable states—either a stable limit cycle or a stable fixed point—each
surrounded by the same unstable limit cycle $\Gamma_{\SP}$.

To keep the discussion focused on the mechanism relevant for the present work, we do not use the full detailed asymptotic expansions available in the preceding papers \cite{berglund2004noise,berglund2014noise,maier1996oscillatory}, including refinements related to cycling,  exit-location statistics, and oscillatory modulation. Instead, we assume only the common leading-order weak-noise structure of the mean first-exit times, namely an exponential term of Arrhenius type combined with a prefactor having polynomial dependence on $\sigma$. This generic exponential--polynomial ansatz is consistent with the asymptotic results obtained by Berglund--Gentz for escape from a stable limit cycle through an unstable limit cycle \cite{berglund2004noise,berglund2014noise}, and by Maier--Stein for escape from a stable fixed point through an unstable limit cycle \cite{maier1996oscillatory}. Accordingly, in the weak-noise regime $0<\sigma\leq\sigma_0$ with $0<\sigma_0\ll1$, we postulate
\begin{align}\label{eq:mean_exits}
\begin{split}
\left\{\begin{array}{lcl}
\mathbb{E}[\tau_\FP]
&=&
P_{\FP}\,\sigma^{\beta_\FP}
\exp\!\left(\frac{S_{\FP}}{\sigma^{2}}\right),\\
\mathbb{E}[\tau_\LC]
&=&
P_{\LC}\,\sigma^{\beta_\LC}
\exp\!\left(\frac{S_{\LC}}{\sigma^{2}}\right).
\end{array}\right.
\end{split}
\end{align}
Here, $P_{\FP}=P_{\FP}(a,\varepsilon)$ and $P_{\LC}=P_{\LC}(a,\varepsilon)$ are positive model-dependent coefficients, $\beta_\FP=\beta_\FP(a,\varepsilon)$ and $\beta_\LC=\beta_\LC(a,\varepsilon)$ describe the polynomial dependence of the prefactors on the noise amplitude $\sigma$, and $S_{\FP}=S_{\FP}(a,\varepsilon)$ and $S_{\LC}=S_{\LC}(a,\varepsilon)$ are the quasi-potentials given by Eq. \eqref{eq:barriers-sigma}. The ansätze in Eq.~\eqref{eq:mean_exits} are not tied to the specific model considered here.
Consequently, the competition between the two escape processes, and hence the occurrence
of ISR, is not model-specific. Instead, it reflects a structurally robust
small-noise mechanism that becomes observable on sufficiently long timescales.

In the small-noise and bi-stable regimes, trajectories mix rapidly within each basin and only rarely escape across the separatrix. Consequently, the system exhibits a clear time-scale separation: the stochastic trajectories remain for long random times within the neighborhood of either the fixed point $X_\FP$ or the limit cycle $\Gamma_\LC$ before noise-induced transitions through the unstable limit cycle $\Gamma_\SP$ occur between their respective basins of attraction.

According to the exponential exit law \cite[Corollary~2]{day1983exponential} 
, both exit times, $\tau_\FP$ and $\tau_\LC$ are asymptotically exponentially distributed in the small-noise limit:
\begin{equation}
\frac{\tau_\FP}{\mathbb{E}[\tau_\FP]}
\;\xrightarrow{d}\;
\mathrm{exp}(1),
\qquad
\frac{\tau_\LC}{\mathbb{E}[\tau_\LC]}
\;\xrightarrow{d}\;
\mathrm{exp}(1),
\qquad
\text{as } \sigma \to 0. 
\end{equation}
Consequently, their survival probabilities decay exponentially:
\begin{align}\label{eq:Floquet_Lyapunovv}
\begin{split}
\left\{\begin{array}{lcl}
\mathbb{P}\{\tau_\FP > t\}
\approx
\exp\!\left(-\frac{t}{\mathbb{E}[\tau_\FP]}\right),\\
\mathbb{P}\{\tau_\LC > t\}
\approx
\exp\!\left(-\frac{t}{\mathbb{E}[\tau_\LC]}\right).
\end{array}\right.
\end{split}
\end{align}
These relations establish a direct correspondence between the mean first-exit times and the transition rates that govern the effective two-state Markov dynamics, a reduction that we argue for in the following.

We fix a small constant $\delta>0$. 
Let $\mathcal{B}^{\delta}_{\LC}$ and $\mathcal{B}^{\delta}_{\FP}$ be  closed $\delta$-neighborhoods of $\Gamma_\LC$ and $X_{\FP}$, respectively. Let $X_t$ be the  solution to Eq. \eqref{eq:2} with $X_0 \in  \{X_{\FP}\} \cup  \Gamma_\LC $. Let  $\tau^{\leftarrow }_0:=0$ and $\tau^{ \to}_{1}$ be the first exit time from $\mathcal{B}_{\FP}$, \textit{i.e.,} 
\begin{equation}
\tau^{ \to}_{1}:= \inf \{t>0: X_t \in \Gamma_\SP \}\,.    
\end{equation}
Then we recursively define the stopping times
\begin{align}
\begin{split}
\left\{\begin{array}{lcl}
\tau_k^{\leftarrow}&:=& \inf \{t>\tau^{ \to}_k: X_t \in \mathcal{B}^{\delta}_{\FP}\cup \mathcal{B}^{\delta}_{\LC} \},\\ 
\tau^{ \to }_{k+1}&:=& \inf \{t>\tau^{ \leftarrow }_{k}: X_t \in  \Gamma_\SP\}.
\end{array}\right.
\end{split}
\end{align}

Now we introduce a random variable $Y_k$ that indicates to which basin of attraction we return after the $k$-th exit to the separatrix, namely:
\begin{equation}
Y_k:= \begin{cases}
\mathrm{FP} & \text{ if } X_{\tau_k^{\leftarrow}} \in \mathcal{B}^{\delta}_{\FP}\,,
\\
\mathrm{LC} & \text{ if } X_{\tau_k^{\leftarrow}} \in \mathcal{B}^{\delta}_{\LC}\,.
\end{cases}    
\end{equation}
Compared to the exponentially large (in $\sigma^{-2}$) exit times $\tau^{\to}_{k+1}-\tau^{\leftarrow }_{k}$ from one of the basins of attraction, the time to reach the neighborhoods $\mathcal{B}^{\delta}_{\FP}\cup \mathcal{B}^{\delta}_{\LC}$ of the fixed point or limit cycle is very short. In fact we expect $\tau^{\leftarrow }_{k}-\tau^{ \to }_{k}$ to be of order $-\log \sigma$ \cite{berglund2006noise,kuehn2015multiple}. Therefore the waiting time between consecutive visits to $\mathcal{B}^{\delta}_{\FP}\cup \mathcal{B}^{\delta}_{\LC}$ is dominated by the exits times. We therefore assume that these waiting times 
\begin{equation}
\tau_k:= \tau^{\leftarrow }_{k}-\tau^{\leftarrow }_{k-1}\,,    
\end{equation}
 become independent in the small noise regime and their distribution only depends on the starting point and approaches the distribution of $\tau_{\#}$, when conditioned on $Y_{k-1} = \#$ for $\#\in\{\mathrm{FP},\mathrm{LC}\}$. Since starting from the separatrix, it is equally likely to next visit either basin of attraction, the $Y_k$ are asymptotically independent and uniformly distributed on the two possible states. 

Now let $Y_t$ be the continuous time jump process on $\{\mathrm{FP},\mathrm{LC}\}$ that takes the value $Y_k$ on the time interval $[\tau^{\leftarrow }_{k},\tau^{\leftarrow }_{k+1})$. This process indicates in which basin of attraction $X_t$ currently is. Due to the preceding argument in the small noise limit we expect that $Y_t$ is well approximated by the jump process $\widehat{Y}_t$ on $\{\mathrm{FP},\mathrm{LC}\}$ defined as follows.
 
 Let $\widehat{Y}_k$ be i.i.d. on $\{\mathrm{FP},\mathrm{LC}\}$ uniformly distributed random variables for $k\ge 0$ and $\widehat{Y}_0:=Y_0$. We pick random variables $\widehat{\tau}_k$ independent of $(\widehat{Y}_l)_{l\ne k-1}$ and of each other, such that the conditional distribution of  $\widehat{\tau}_k$ given the event $\widehat{Y}_{k-1}= y$ is $\mathrm{Exp}(1/\mathbb E[\tau_{y}])$ and $\widehat \tau^{\leftarrow }_{k}:=\sum_{l=1}^k\widehat{\tau}_l$. Then $\widehat{Y}_t$ takes the value $\widehat{Y}_k$ for $t\in [\widehat \tau^{\leftarrow }_{k},\widehat  \tau^{\leftarrow }_{k+1})$.
 
 The process $\widehat{Y}_t$ starts at $\widehat{Y}_0$. Once entering a state $\widehat{Y}_k=y$ it stays in this state for a $\mathrm{Exp}(1/\mathbb E[\tau_{y}])$-distributed waiting time. After that it jumps to a state  $\widehat{Y}_{k+1}$ that is uniformly distributed on $\{\mathrm{FP},\mathrm{LC}\}$. For the continuous-time Markov process $\widehat{Y}_t$, it is easy to see that 
 \begin{equation}\label{prob_LC}
\mu_\sigma(\mathcal{B}^{\delta}_{\LC})=\frac{1}{T} \int _0^T \mathbbm{1}_{\{\widehat{Y}_t= \LC\}} d t  \to \frac{\mathbb E[\tau_{\LC}]}{\mathbb E[\tau_{\LC}]+\mathbb E[\tau_{\FP}]}, 
 \end{equation}
 almost surely as $T \to \infty$. 

In the weak-noise regime $0<\sigma\leq\sigma_0$ with $0<\sigma_0\ll1$, the following asymptotic approximation holds for the
stationary occupation probability of $\mathcal{B}^{\delta}_{\LC}$:
\begin{equation}\label{eq:muBLC_refined}
\mu_\sigma(\mathcal{B}^{\delta}_{\LC})
\;=\;
\frac{1}{1+\mathbb E[\tau_{\FP}]/\mathbb E[\tau_{\LC}]}
\;=\;
\left[
1+
P\,\sigma^{\beta}
\exp\!\left(
-\frac{\Delta S}{\sigma^{2}}
\right)
\right]^{-1}.
\end{equation}
Now let
\begin{equation}\label{R}
R(\sigma):=
P\,\sigma^{\beta}
\exp\!\left(
-\frac{\Delta S}{\sigma^{2}}
\right),
\end{equation}
where $P>0$, $\Delta S\neq0$, and $\beta\neq0$, so that
Eq.~\eqref{eq:muBLC_refined} can be written as
\begin{equation}\label{RR}
\mu_\sigma(\mathcal{B}^{\delta}_{\LC})=\frac{1}{1+R(\sigma)}.
\end{equation}
When $R$ is large, $\mu_\sigma(\mathcal{B}^{\delta}_{\LC})$ is close to $0$;
when $R$ is small, $\mu_\sigma(\mathcal{B}^{\delta}_{\LC})$ is close to $1$.
Differentiating with respect to $\sigma$ gives
\begin{equation}\label{eq:dmu_refined}
\frac{d}{d\sigma}\mu_\sigma(\mathcal{B}^{\delta}_{\LC})
=
-\frac{R'(\sigma)}{(1+R(\sigma))^{2}},
\qquad
R'(\sigma)
=
R(\sigma)\left(
\frac{\beta}{\sigma}
+
\frac{2\Delta S(a,\varepsilon)}{\sigma^{3}}
\right).
\end{equation}
An interior turning point exists when
$\beta\sigma^2+2\Delta S=0$. Thus the corresponding extremal noise amplitude
$\sigma_\ast>0$ is given by
\begin{equation}\label{predicted_noise}
 \sigma_\ast^2
=
-\frac{2\Delta S}{\beta}>0.
\end{equation}
Hence, an interior turning point exists precisely when $\Delta S$ and $\beta$
have opposite signs.

The signs of $\Delta S$ and $\beta$ determine the qualitative shape of
$\mu_\sigma(\mathcal B^\delta_{\LC})$ within the weak-noise regime in which
Eq.~\eqref{eq:muBLC_refined} is valid. 
\begin{itemize}
\item Case $\Delta S>0,\ \beta<0$:
For weak noise, $R(\sigma)\to0$ as $\sigma\to0^+$, and hence
$\mu_\sigma(\mathcal B^\delta_{\LC})\to1$. Since $\Delta S\beta<0$,
there exists a positive critical noise amplitude
$
\sigma_\ast=\sqrt{-2\Delta S/\beta}.
$
At this point $R$ has a maximum, since $\beta<0$, and therefore
$\mu_\sigma(\mathcal B^\delta_{\LC})$ has a minimum. If
$\sigma_\ast\in(0,\sigma_0)$, then
$
\mu_{\sigma_\ast}(\mathcal B^\delta_{\LC})
=
\min_{0<\sigma\leq\sigma_0}
\mu_\sigma(\mathcal B^\delta_{\LC}).
$
Thus, within the weak-noise regime, $\mu_\sigma(\mathcal B^\delta_{\LC})$
has the ISR-type local profile
$
1\;\searrow\;\mu_{\sigma_\ast}(\mathcal B^\delta_{\LC})\;\nearrow .
$
This is the genuine ISR dip, provided $\sigma_\ast$ lies in the interior of
the weak-noise range where Eq.~\eqref{eq:muBLC_refined} is applicable.
\item Case $\Delta S>0,\ \beta>0$:
For weak noise, $R(\sigma)\to0$ as $\sigma\to0^+$, and hence
$\mu_\sigma(\mathcal B^\delta_{\LC})\to1$. Since $\Delta S$ and $\beta$
have the same sign, Eq.~\eqref{predicted_noise} has no positive solution.
Moreover, $R'(\sigma)>0$ for all $\sigma>0$, so $R$ is monotone increasing
and $\mu_\sigma(\mathcal B^\delta_{\LC})$ is monotone decreasing in the
regime described by Eq.~\eqref{eq:muBLC_refined}. Hence the local weak-noise
profile is
$
1\;\searrow .
$
This case does not produce an ISR minimum.
\item Case $\Delta S<0,\ \beta<0$:
For weak noise, $R(\sigma)\to\infty$ as $\sigma\to0^+$, and hence
$\mu_\sigma(\mathcal B^\delta_{\LC})\to0$. Since $\Delta S$ and $\beta$
have the same sign, Eq.~\eqref{predicted_noise} has no positive solution.
Furthermore, $R'(\sigma)<0$ for all $\sigma>0$, so $R$ is monotone decreasing
and $\mu_\sigma(\mathcal B^\delta_{\LC})$ is monotone increasing in the
weak-noise regime. Thus the local profile is
$
0\;\nearrow .
$
This case does not produce an ISR minimum.
\item Case $\Delta S<0,\ \beta>0$:
For weak noise, $R(\sigma)\to\infty$ as $\sigma\to0^+$, and hence
$\mu_\sigma(\mathcal B^\delta_{\LC})\to0$. Since $\Delta S\beta<0$, there
exists a positive critical noise amplitude
$
\sigma_\ast=\sqrt{-2\Delta S/\beta}.
$
At this point $R$ has a minimum, since $\beta>0$, and therefore
$\mu_\sigma(\mathcal B^\delta_{\LC})$ has a maximum. If
$\sigma_\ast\in(0,\sigma_0)$, then
$
\mu_{\sigma_\ast}(\mathcal B^\delta_{\LC})
=
\max_{0<\sigma\leq\sigma_0}
\mu_\sigma(\mathcal B^\delta_{\LC}).
$
Thus, within the weak-noise regime, the occupation probability has the
peak-shaped local profile
$
0\;\nearrow\;\mu_{\sigma_\ast}(\mathcal B^\delta_{\LC})\;\searrow .
$
This is the mirror case of the ISR dip and does not correspond to spike suppression by an intermediate noise amplitude.
\end{itemize}

The key observation is that non-monotonicity can occur only when
$\Delta S\beta<0$. More precisely,
$\Delta S>0,\ \beta<0$ yields a dip in $\mu_\sigma$, whereas
$\Delta S<0,\ \beta>0$ yields a peak in $\mu_\sigma$. Thus, the ISR dip is not
a consequence of $\Delta S>0$ alone; it results from the competition between the
exponential factor $\exp\!\left(-\Delta S/\sigma^2\right)$ and the power-law
prefactor $\sigma^\beta$ with negative effective exponent $\beta<0$.

We now test these criteria numerically for the FHN model defined by Eq.~\eqref{eq:2}. 
All numerical quantities used in Figs.~\ref{fig:9}--\ref{fig:11} are obtained as follows.  The quasi-potential barriers $S_{\FP}$ and $S_{\LC}$ are computed from Eq.~\eqref{eq:barriers-sigma} using the degenerate-noise gMAM method; hence $\Delta S=S_{\LC}-S_{\FP}$.
The occupation probability $\mu_\sigma(\mathcal B_{\LC})$ is estimated as per Eq. \eqref{eq:mfr_proxy} by
Monte Carlo simulation of the stochastic FHN system $X_t=(v_t,,w_t)$ in Eq. \eqref{eq:2} using
Euler--Maruyama method.  Similarly, the mean first-exit times entering Eq.~\eqref{eq:log_ratio_refined} are estimated by Monte Carlo averages of the first-exit times from
$\mathcal B_{\FP}$ and $\mathcal B_{\LC}$, respectively.

Figure~\ref{fig:9} shows how the quasi-potential imbalance
$\Delta S$ organizes the bistable wedge in parameter space. 
The balance curve
$\Delta S=0$ separates two escape-dominated regimes. In the region
$\Delta S>0$, escape from the limit-cycle basin is exponentially more costly
than escape from the fixed-point basin; for a negative effective exponent $\beta$, the
reduced theory therefore predicts an ISR minimum. 
In the region $\Delta S<0$, the quasi-potential ordering is reversed. For the parameter points considered here, fitting Eq.~\eqref{eq:log_ratio_refined} over the weak-noise interval ($\log \sigma \in [-6.0,-3.3]$), whose choice is justified below, yields a negative effective exponent $(\beta<0)$. In this case, the reduced theory predicts no ISR dip; if instead ($\beta>0$), the non-monotonicity would be peak-shaped rather than an ISR minimum. The marked parameter points are used below to test these regimes numerically.
\begin{figure}
\centering
\includegraphics[width=8.5cm,height=6.0cm]{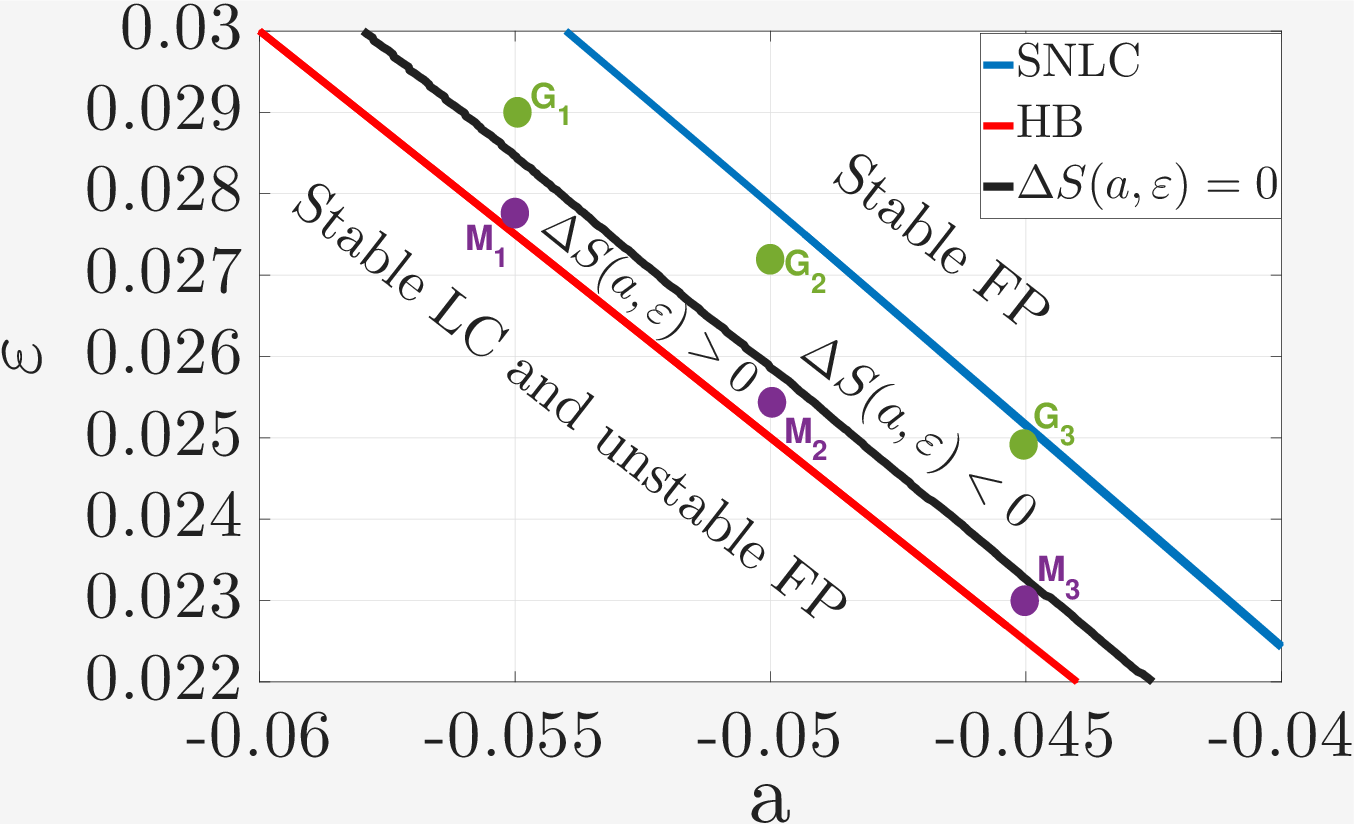}
\caption{Quasi-potential partition of the bistable region in the $(a,\varepsilon)$-plane. The black curve is the numerically computed balance set $\Delta S(a,\varepsilon)=S_{\LC}-S_{\FP}=0$, separating the regions $\Delta S>0$ and $\Delta S<0$. The red and blue curves denote the Hopf and SNLC boundaries, respectively. Magenta markers denote test points with $\Delta S>0$: $M_1=(-0.055,0.02751)$, $M_2=(-0.050,0.0253784)$, and $M_3=(-0.045,0.023313539824)$. Green markers denote test points with $\Delta S<0$: $G_1=(-0.055,0.029)$, $G_2=(-0.050,0.0272325)$, and $G_3=(-0.045,0.0250727)$. The magenta and green points are used in Figs.~\ref{fig:10} and \ref{fig:11}, respectively. $b=1.0$ and $c=2.0$.}
\label{fig:9}
\end{figure}

It is useful to rewrite the ratio of mean escape times in logarithmic form.
From Eq.~\eqref{eq:muBLC_refined}, we obtain
\begin{equation}\label{eq:log_ratio_refined}
\log\!\left(\frac{\mathbb E[\tau_{\FP}]}{\mathbb E[\tau_{\LC}]}\right)
=
\log P
+
\beta\log \sigma
-
\frac{\Delta S}{\sigma^{2}}.
\end{equation}
Equation~\eqref{eq:log_ratio_refined} separates the algebraic prefactor
contribution from the Arrhenius-type exponential contribution. 

Figure~\ref{fig:10} shows fits of Eq.~\eqref{eq:log_ratio_refined} to the data and illustrates the criterion for the three magenta test points ($M_1,M_2,M_3$) in Fig.~\ref{fig:9}, all satisfying $\Delta S>0$. Panels (a1)--(c1) show the Monte Carlo estimates of $\mu_\sigma(B_{\mathrm{LC}})$, while panels (a2)--(c2) show the corresponding logarithmic escape-time ratios after subtracting the leading quasi-potential contribution $(-\Delta S/\sigma^2)$. The resulting transformed quantity is used to estimate the effective exponent $\beta$ in  Eq.~\eqref{eq:log_ratio_refined} over the weak-noise fitting interval. The fitted values of $\beta$ are then used to evaluate the reduced escape-balance criterion for an ISR minimum. This provides a numerical consistency check for the reduced escape-balance description.

\begin{figure}[h]
\centering 
\includegraphics[width=5.0cm,height=9.0cm]{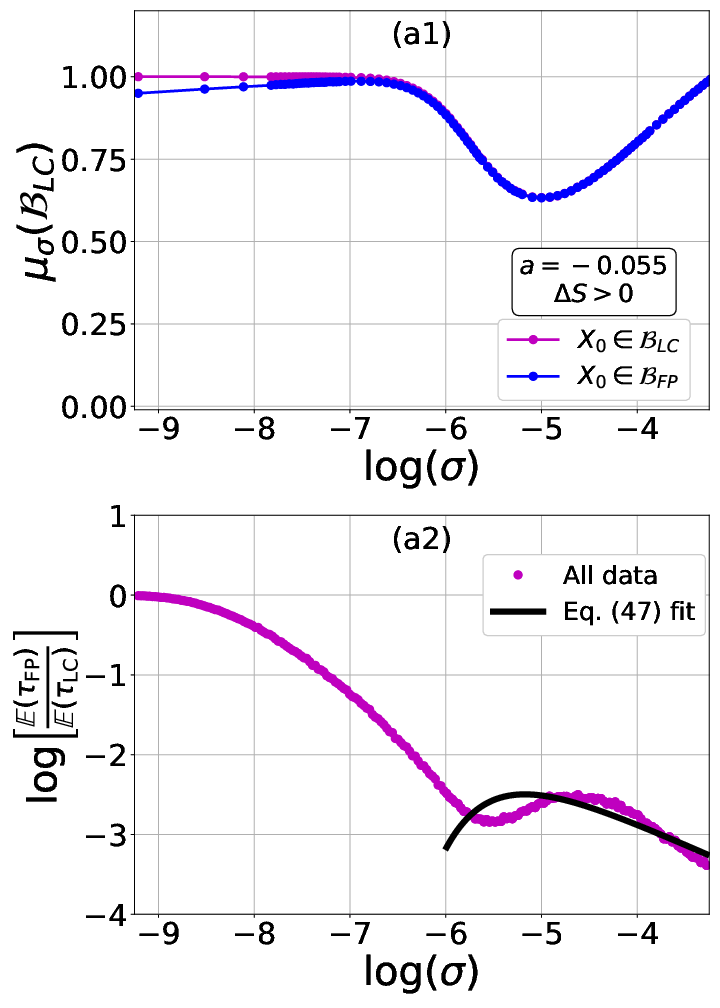}
\includegraphics[width=5.0cm,height=9.0cm]{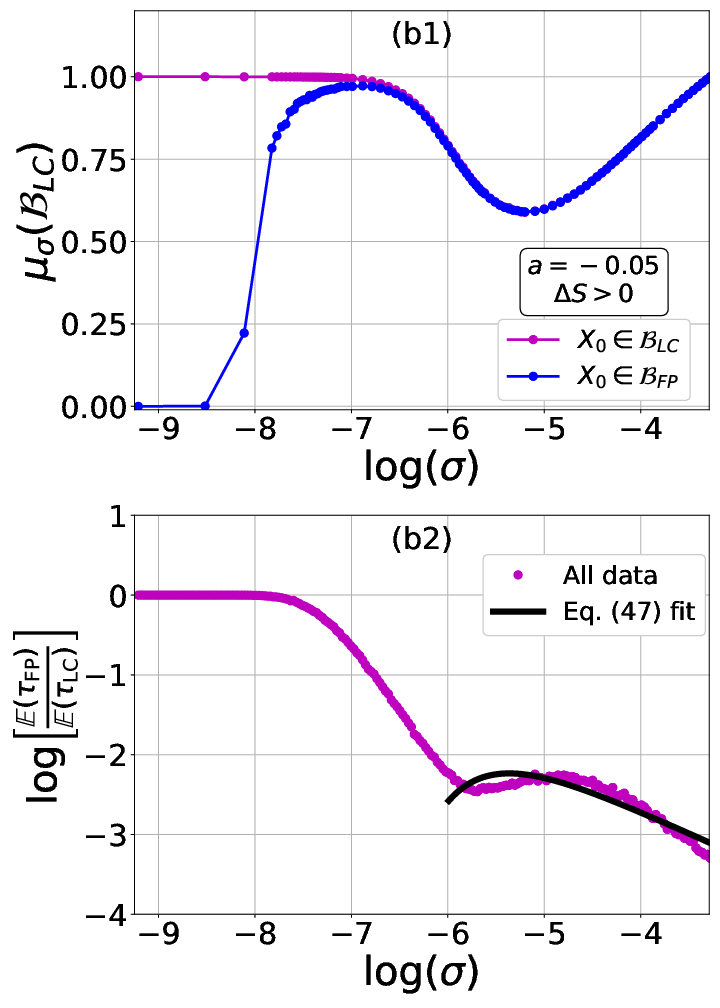}
\includegraphics[width=5.0cm,height=9.0cm]{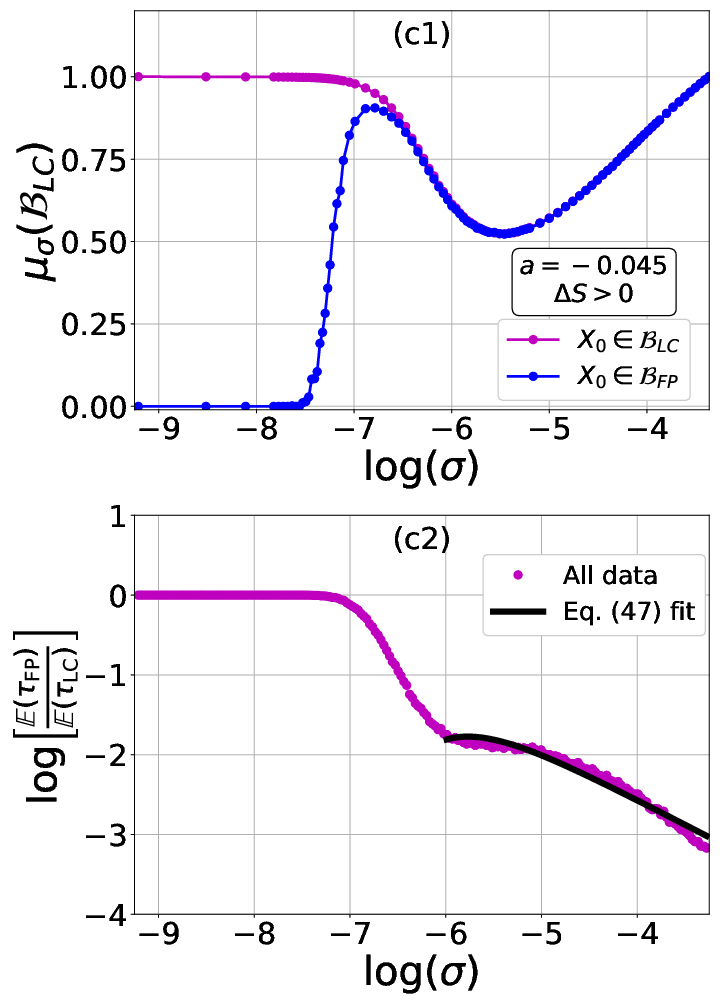}
\caption{
Numerical validation of the weak-noise ISR criterion for the three magenta parameter points \(M_1,M_2,M_3\) in Fig.~\ref{fig:9}, all satisfying $\Delta S>0$. Panels (a1)--(c1) show Monte Carlo estimates of $\mu_\sigma(\mathcal B_{\LC})$ as functions of $\log\sigma$, with the observed minimum denoted by $\sigma_{\min}$. Panels (a2)--(c2) show the transformed logarithmic escape-time ratios used to estimate the effective exponent $\beta$ in Eq.~\eqref{eq:log_ratio_refined}. The fits are performed over the weak-noise interval $\log\sigma\in[-6.0,-3.3]$. Columns (a), (b), and (c) correspond to \(M_1\), \(M_2\), and \(M_3\), respectively. Other parameters are $T=2.5\times 10^{5}$, $b=1.0$, and $c=2.0$.
}
\label{fig:10}
\end{figure}

For each parameter point, the quasi-potential difference $\Delta S$ was first computed independently from the quasi-potential landscape using the degenerate-noise gMAM and then kept fixed during the fitting procedure, so that only $\log P$ and $\beta$ were estimated from the mean escape-time data.
Rewriting Eq.~\eqref{eq:log_ratio_refined} as
\begin{equation}
\label{eq:linear_fit_corrected}
    y+\frac{\Delta S}{\sigma^2}=C+\beta x,
    \qquad
    y=\log\!\left(\frac{\mathbb E[\tau_{\FP}]}
    {\mathbb E[\tau_{\LC}]}\right),\quad
    x=\log\sigma,\quad C=\log P,
\end{equation}
reduces the estimation of the remaining parameters to a linear regression
problem. The effective prefactor parameters $C=\log P$ and $\beta$ are then
obtained by ordinary least-squares fitting of the corrected data
$y+\Delta S/\sigma^2$ against $x$ over the fitting interval used in
Fig.~\ref{fig:10}.

The fitting interval $\log\sigma\in[-6.0,-3.3]$ was chosen to isolate the noise range in which the metastable two-state reduction is both numerically resolved and dynamically relevant for the ISR transition. For very weak noise amplitudes, $\log \sigma\in[-9.2,-6)$, exits from the more stable basin are too rare within our long but finite simulation horizon $T$. As a result, the mean exit times saturate at $T$, \textit{i.e.,} $\mathbb E[\tau_{\FP}] = \mathbb E[\tau_{\LC}] = T,
$ which gives $\log\left(\mathbb E[\tau_{\FP}]/\mathbb E[\tau_{\LC}]\right)=\log(1)=0$ (see Figs. \ref{fig:10}(a2)-(c2) and \ref{fig:11}(a2)-(c2) for $\log \sigma\in[-9.2,-6)$). In this regime, the estimated occupation probability $\mu_\sigma(\mathcal B_{\LC})$ also shows a visible dependence on the initial basin, either $\mathcal B_{\LC}$ or $\mathcal B_{\FP}$. Therefore, the data in the interval $\log \sigma\in[-9.2,-6)$ are excluded from the fit.

In contrast, in the rest of the interval,\textit{i.e.,} $\log\sigma\in[-6.0,-3.3]$, the occupation probability curves obtained from initial conditions in $\mathcal B_{\FP}$ and $\mathcal B_{\LC}$ largely collapse, within sampling accuracy, onto the same $\mu_{\sigma}(\mathcal B_{\LC})$ profile for the parameter sets used in
the fits. This indicates that transitions between the two basins are
sufficiently well sampled and that the estimates are representative of the
invariant occupation measure. The upper end of the interval, $\log\sigma=-3.3$, is chosen to remain below the regime where moderate-noise effects distort the weak-noise
exponential--algebraic scaling in Eq.~\eqref{eq:log_ratio_refined}. Thus, the
selected interval, $\log\sigma\in[-6.0,-3.3]$, provides a practical compromise between avoiding unresolved rare-event artefacts and hence dependence of $\mu_\sigma(\mathcal B_{\LC})$ on the initial basin at too small noise amplitudes and avoiding non-asymptotic corrections at larger noise $\log \sigma>-3.3$.

Table~\ref{tab:sigma_comparison} summarizes the resulting quantitative test for
the three parameter values indicated by magenta markers  ($M_1,M_2,M_3$) in Fig.~\ref{fig:9}.
For all three points, the fitted effective exponent is negative, so that
$\Delta S\beta<0$ and Eq.~\eqref{predicted_noise} yields a positive
ISR-minimizing noise amplitude. The predicted values $\sigma_\ast$ are close to
the observed minima $\sigma_{\min}$ measured from the Monte Carlo occupation
curves. The remaining discrepancies are expected, since
Eq.~\eqref{predicted_noise} retains only the leading exponential--algebraic
structure of the mean-exit times and neglects higher-order prefactor
corrections, finite-time sampling effects, and non-asymptotic deviations at
moderate noise. Thus, the table supports the proposed mechanism: in the regime
$\Delta S>0$ and $\beta<0$, the ISR minimum is governed by the competition
between the quasi-potential imbalance and the algebraic noise-dependence of the
escape-time prefactors.

\begin{table}
\centering
\scriptsize
\setlength{\tabcolsep}{3pt}
\caption{
Quantitative comparison between the predicted and observed ISR minima for the three magenta parameter points \(M_1,M_2,M_3\) in Fig.~\ref{fig:9}, all satisfying $\Delta S>0$. The columns report the gMAM-computed barrier difference $\Delta S$, the fitted parameters $\log P$ and $\beta$ in Eq.~\eqref{eq:log_ratio_refined}, the predicted value $\sigma_\ast$ from Eq.~\eqref{predicted_noise}, the observed minimum $\sigma_{\min}$ of $\mu_\sigma(\mathcal B_{\LC})$, the absolute prediction error, and the relative error.
}
\label{tab:sigma_comparison}
\begin{tabular*}{\linewidth}{@{\extracolsep{\fill}}lccccccc@{}}
\hline
Test point
& $\Delta S$
& $\log P$
& $\beta$
& $\sigma_\ast$
& $\sigma_{\min}$
& $|\sigma_\ast-\sigma_{\min}|$
& $|\sigma_\ast-\sigma_{\min}|/\sigma_{\min}$ \\
\hline
$M_1$
& $8.5470\times10^{-6}$
& $-4.99$
& $-0.5331$
& $0.5663\times 10^{-2}$
& $0.6733\times 10^{-2}$
& $0.0011$
& $0.1589$\\
$M_2$
& $6.0123\times10^{-6}$
& $-4.90$
& $-0.5480$
& $0.4684\times10^{-2}$
& $0.5505\times10^{-2}$
& $0.0008$
& $0.1491$\\
$M_3$
& $4.7459\times10^{-6}$
& $-5.34$
& $-0.6935$
& $0.3695\times10^{-2}$
& $0.4042\times10^{-2}$
& $0.0003$
& $0.0940$\\
\hline
\end{tabular*}
\end{table}

Figure~\ref{fig:11} gives the complementary test for the three green points  ($G_1,G_2,G_3$) in
Fig.~\ref{fig:9}, for which $\Delta S<0$. The logarithmic escape-ratio plots in
panels (a2)--(c2) again yield negative fitted effective exponents $\beta$. Hence
$\Delta S$ and $\beta$ have the same sign, so that
Eq.~\eqref{predicted_noise} has no positive real solution. The reduced
escape-balance theory, therefore, predicts that no genuine ISR minimum should
occur for these parameter values. The apparent depressions visible in the
finite-time occupation curves initialized in $\mathcal B_{\LC}$ are consistent
with unresolved rare escapes over the simulation horizon, rather than with an
asymptotic ISR dip.

\begin{figure}[h]
\centering 
\includegraphics[width=5.0cm,height=9.0cm]{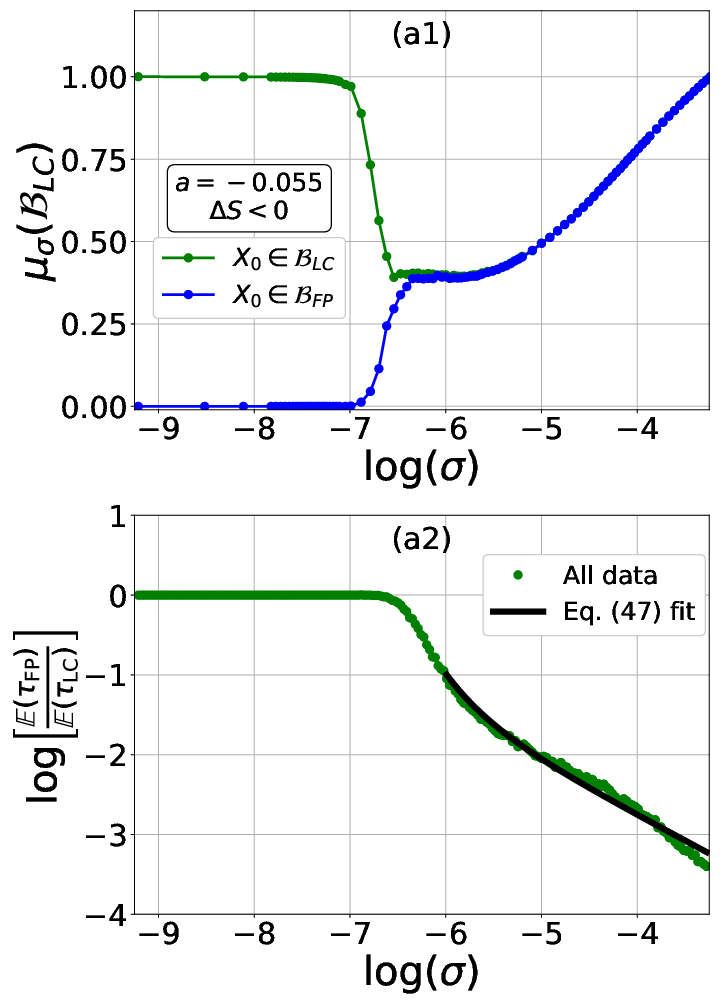}
\includegraphics[width=5.0cm,height=9.0cm]{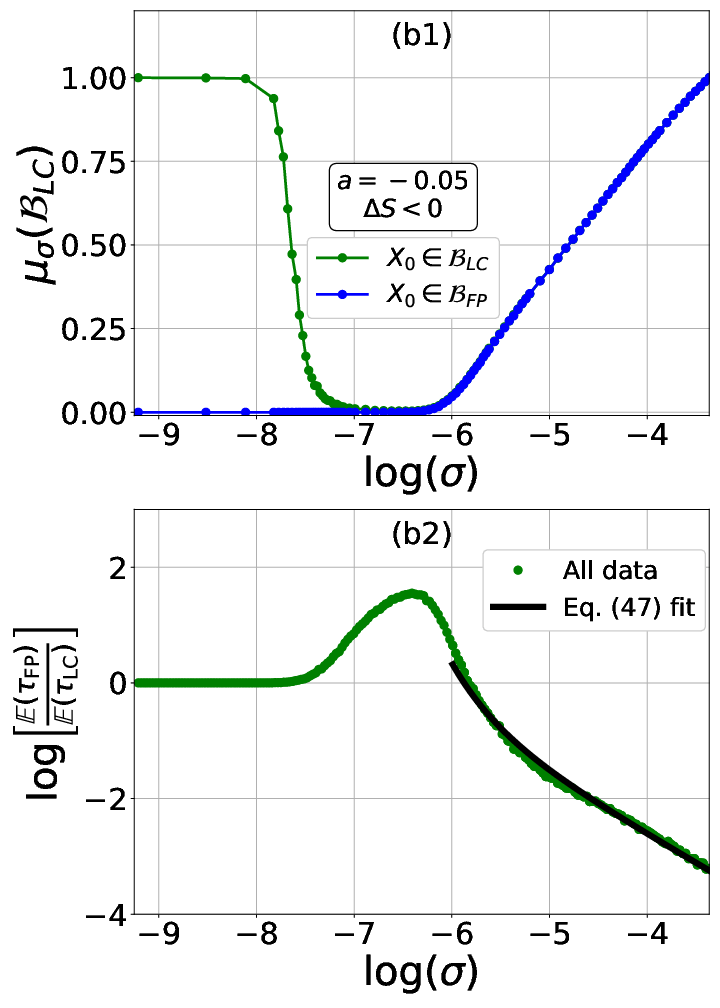}
\includegraphics[width=5.0cm,height=9.0cm]{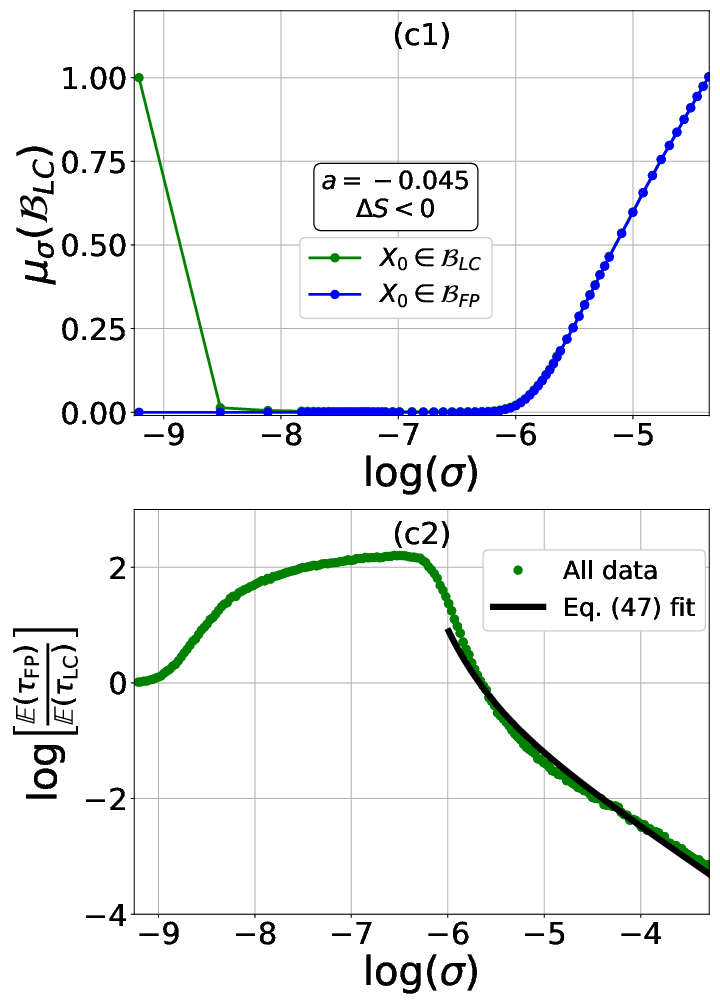}
\caption{
Numerical test of the weak-noise criterion for the three green parameter points \(G_1,G_2,G_3\) in Fig.~\ref{fig:9}, all satisfying $\Delta S<0$. Panels (a1)--(c1) show Monte Carlo estimates of $\mu_\sigma(\mathcal B_{\LC})$ as functions of $\log\sigma$. Panels (a2)--(c2) show the transformed logarithmic escape-time ratios used to estimate the effective exponent $\beta$ in Eq.~\eqref{eq:log_ratio_refined}. The fits are performed over the weak-noise interval $\log\sigma\in[-6.0,-3.3]$. Columns (a), (b), and (c) correspond to \(G_1\), \(G_2\), and \(G_3\), respectively. Since the fitted values of $\beta$ are negative while $\Delta S<0$, Eq.~\eqref{predicted_noise} has no positive real solution, and the reduced theory predicts no genuine asymptotic ISR minimum. Other parameters are $T=2.5\times10^{5}$, $b=1.0$, and $c=2.0$.
}
\label{fig:11}
\end{figure}

Table~\ref{tab:green_no_isr} gives the complementary quantitative check for the
green parameter points. Unlike the magenta cases in
Table~\ref{tab:sigma_comparison}, these points are not used to compare a
predicted minimum with an observed one. Instead, they test the non-existence
part of the weak-noise criterion. For each green point, the right-hand side of
Eq.~\eqref{predicted_noise} is negative, so the critical-noise formula does not
define a positive real noise amplitude. The table therefore supports the
interpretation that the depressions visible (see, \textit{e.g.,} Fig. \ref{fig:11}(a1)-(b1)) in finite-time occupation curves are
not genuine asymptotic ISR minima.

Furthermore, the quality of the fit in Eq.~\eqref{eq:linear_fit_corrected} depends on the location of the parameter point within the bistable wedge. For the magenta points  ($M_1,M_2,M_3$) with $\Delta S>0$, shown in Fig.~\ref{fig:9}, the fit improves as the parameters move away from the Hopf boundary: the root-mean-square error decreases from (0.182) in Fig.~\ref{fig:10}(a2), to (0.116) in Fig.~\ref{fig:10}(b2), and to (0.063) in Fig.~\ref{fig:10}(c2). Conversely, for the green points ($G_1,G_2,G_3$) with $\Delta S<0$, also shown in Fig.~\ref{fig:9}, the fit deteriorates as the parameters approach the SNLC boundary: the root-mean-square error increases from (0.080) in Fig.~\ref{fig:11}(a2), to (0.092) in Fig.~\ref{fig:11}(b2), and to (0.149) in Fig.~\ref{fig:11}(c2).

\begin{table}
\centering
\scriptsize
\setlength{\tabcolsep}{3pt}
\caption{Quantitative test of the non-existence part of the weak-noise criterion for the three green parameter points \(G_1,G_2,G_3\) in Fig.~\ref{fig:9}, all satisfying $\Delta S<0$. The columns report the gMAM-computed barrier difference $\Delta S$, the fitted parameters $\log P$ and $\beta$ in Eq.~\eqref{eq:log_ratio_refined}, and the prediction obtained from Eq.~\eqref{predicted_noise}.
}
\label{tab:green_no_isr}
\begin{tabular*}{\linewidth}{@{\extracolsep{\fill}}lcccl@{}}
\hline
Test point
& $\Delta S$
& $\log P$
& $\beta$
& Prediction \\
\hline
$G_1$
& $-3.0656\times10^{-6}$
& $-5.32$
& $-0.6395$
& No positive real $\sigma_\ast$ \\
$G_2$
& $-6.1528\times10^{-6}$
& $-6.56$
& $-0.9842$
& No positive real $\sigma_\ast$ \\
$G_3$
& $ -6.9879\times10^{-6}$
& $-7.07$
& $-1.1424$
& No positive real $\sigma_\ast$ \\
\hline
\end{tabular*}
\end{table}

These trends are consistent with the quasi-potential structure shown in Fig.~\ref{fig:6}: near the Hopf boundary, the barrier ($S_{\FP}$) for escape from the fixed-point basin is small while $S_{\LC}$ is comparatively large; near the SNLC boundary, the barrier ($S_{\LC}$) for escape from the limit-cycle basin decreases while $S_{\FP}$ is comparatively large. In both cases, one of the two attracting states is less strongly metastable, so finite simulation time, higher-order prefactor effects, and other non-asymptotic corrections become more visible over the accessible noise range. Thus, the variation in fit quality reflects boundary-induced limitations of the leading exponential--algebraic approximation in Eq.~\eqref{eq:log_ratio_refined}, rather than a failure of the escape-balance mechanism. In particular, the deterioration of the fit for the green points near the SNLC boundary does not alter the qualitative prediction:
since $\Delta S<0$ and the fitted effective exponent $\beta$ remains negative,
Eq.~\eqref{predicted_noise} has no positive real solution. The reduced
weak-noise theory, therefore, predicts no genuine asymptotic ISR minimum for
these parameter values.

\section{Summary and concluding remarks}
\label{sec:conclusion}
We have analyzed inverse stochastic resonance in a bistable
FitzHugh--Nagumo neuron driven by additive noise in the voltage variable. The main focus was on
how the excitability parameter $a$ and the timescale separation parameter
$\varepsilon$ shape the transition structure between quiescent and spiking
states. The deterministic bifurcation analysis identifies a codimension-two
bistability wedge in the $(a,\varepsilon)$-plane, bounded by a subcritical Hopf
bifurcation and a saddle-node of limit cycles. Inside this wedge, a stable fixed
point and a stable limit cycle coexist and are separated by an unstable periodic
orbit. This provides the deterministic phase-space geometry required for
noise-induced switching and hence for bistability-driven ISR.

The Monte Carlo simulations show that finite observation times can produce an
apparent dependence of ISR on the initial basin. In particular, occupation
curves initialized in the fixed-point basin and in the limit-cycle basin may
differ when one of the relevant escape processes is too rare to be sufficiently
sampled. This dependence is a finite-time effect. The uniqueness of the
invariant probability measure implies that, for fixed parameters and
$\sigma>0$, the long-time occupation probability of the limit-cycle basin is
independent of the initial condition. Thus, genuine asymptotic ISR is determined
by the parameter-dependent balance of the two escape mechanisms, not by the
initial state.

This balance is quantified by the quasi-potential barriers
$S_{\FP}(a,\varepsilon)$ and $S_{\LC}(a,\varepsilon)$ for escape from the
fixed-point and limit-cycle basins. The degenerate-noise gMAM computations show
that the difference
\[
    \Delta S(a,\varepsilon)=S_{\LC}(a,\varepsilon)-S_{\FP}(a,\varepsilon)
\]
organizes the bistable wedge into two regimes. When $\Delta S>0$, escape from
the limit-cycle basin is exponentially more costly than escape from the
fixed-point basin. When $\Delta S<0$, the quasi-potential ordering is reversed.
This partition explains why some finite-time non-monotone occupation curves
correspond to genuine ISR, whereas others arise from unresolved rare-event
statistics over the simulation horizon.

A reduced metastable description, based on the leading weak-noise form of the
mean escape times, gives the approximation
\[
    \mu_\sigma(\mathcal B^\delta_{\LC})
    =
    \left[
    1+
    P\sigma^\beta
    \exp\!\left(-\frac{\Delta S}{\sigma^2}\right)
    \right]^{-1}.
\]
This formula yields the sign condition $\Delta S\beta<0$ for
non-monotonicity of the basin occupation. In the fitted regime considered here,
where $\beta<0$, a genuine ISR dip is predicted only on the side
$\Delta S>0$. In that case, the reduced model also gives the predicted minimizing noise amplitude
\[
    \sigma_\ast^2=-\frac{2\Delta S}{\beta}.
\]
The test points with $\Delta S>0$ satisfy this criterion and yield
semiquantitative agreement between the predicted $\sigma_\ast$ and the observed
Monte Carlo minima. The test points with $\Delta S<0$ have no positive
predicted minimizer and therefore do not correspond to genuine asymptotic ISR
within the reduced weak-noise theory. 
The weak-noise approximation for
$\mu_\sigma(\mathcal B^\delta_{\LC})$ also suggests a limiting notion of complete noise-induced suppression of oscillatory-basin occupation near the quasi-potential balance boundary. More precisely, consider a two-dimensional dynamical system depending on a continuous parameter $q$, and suppose that, for
$q\in\mathcal Q$, the deterministic system possesses a stable fixed point
$X_{\FP}(q)$ and a stable limit cycle $\Gamma_{\LC}(q)$ separated by an
unstable limit cycle $\Gamma_{\SP}(q)$. Assume further that
$
\Delta S(q):=S_{\LC}(q)-S_{\FP}(q)>0
$
for $q\in\mathcal Q$, and that $\beta(q)<0$ in this regime. If
$q_0\in\partial\mathcal Q$ is such that $\Delta S(q)\to0^+$ as
$q\to q_0$ from within $\mathcal Q$, then the predicted minimizing noise
amplitude satisfies
$
\sigma_\ast(q)^2=-2\Delta S(q)/\beta(q)\to0.
$
Moreover, provided the prefactor ratio $P(q)$ remains bounded away from zero
and infinity, and $\beta(q)$ remains in a compact subset of $(-\infty,0)$,
the corresponding minimum of the weak-noise approximation satisfies
$
\mu_{\sigma_\ast(q)}^{(q)}(\mathcal B^\delta_{\LC}(q))\to0.
$
Thus the ISR dip enters arbitrarily small noise windows and becomes arbitrarily
deep as $q$ approaches $q_0$.

This motivates the following definition. We say that the 2-d dynamical system exhibits \textit{full
inverse stochastic resonance} at $q_0$ if, for every sequence $q_n\to q_0$ with
$q_n\in\mathcal Q$, and for every sufficiently small $\sigma_0>0$,
\[
\lim_{n\to\infty}
\inf_{0<\sigma\leq\sigma_0}
\mu_\sigma^{(q_n)}(\mathcal B^\delta_{\LC}(q_n))
=
0.
\]
In this sense, full ISR means that, arbitrarily close to $q_0$, the limit-cycle basin occupation probability can be made arbitrarily small by a noise amplitude that lies in any prescribed weak-noise neighborhood of $\sigma=0$.

Overall, the results provide a theoretical mechanism for ISR in a
bistable excitable system: intrinsic parameters modify the deterministic
bifurcation geometry, this geometry determines the quasi-potential barriers,
and the barrier imbalance controls the long-time occupation of the spiking
basin. The framework also distinguishes finite-time ISR-like depressions from
genuine asymptotic ISR, which is essential when interpreting simulations in
metastable regimes with exponentially long residence times. Although developed for the FitzHugh-Nagumo model with degenerate noise, the methodology should be applicable more broadly to excitable systems with coexisting stable equilibria and periodic orbits.

\section*{Acknowledgements} This work was funded by the Department of Data Science (DDS), Friedrich-Alexander-Universit\"at Erlan\-gen-Nürnberg, Germany, and the Deutsche Forschungsgemeinschaft (DFG, German Research Foundation) via the grant YA 764/1-1 to M.E.Y—Project No. 456989199.


\section*{Appendix: Proof of Theorem \ref{thm:FHN-ergodic}}
\label{appendix}
The proof of Theorem~\ref{thm:FHN-ergodic} is essentially contained in the works \cite{Wu2001,leon2018hypoelliptic} which are based on the H\"ormander's subellipticity estimates \cite{Hormander1967} and the uniqueness criterion of Downs, Meyn and Tweedy 
\cite{down1995exponential}. These general results are applied to the concrete stochastic FHN neuron equation studied here. The main new insight is the use of a natural Lyapunov function for the confinement of the dynamics.

\subsection{Local existence and uniqueness}
\label{sec-LocalExist}
In this section, it will be shown that the SDE Eq.~\eqref{eq:2} has a unique strong solution up to the explosion time. The drift coefficients in Eq. \eqref{eq:2} are $f$ and $g$ so that the deterministic vector field is a smooth and locally Lipschitz function $(v,w)\mapsto(f(v,w),g(v,w))$ that has at most
polynomial growth. The diffusion coefficient is the constant vector $(\sigma,0)^{\!\top}$. Define the exit time
\[
  \tau_{R}:=\inf\{t\ge 0:\|X_{t}\|\ge R\}.
\]
For $R$ sufficiently large, one has $\|X_{0}\|<R$. Then introduce
$
f_R(v,w)=
f(-R,w)\mathbbm{1}_{\{v< -R\}}
+
f(v,w)\mathbbm{1}_{\{|v|\leq R\}}
+
f(R,w)\mathbbm{1}_{\{v> R\}}
$
The vector field $(f_R,g)$ is then globally Lipschitz and linearly bounded in $v$ and $w$. Hence, Picard-iteration argument ({\it e.g.,} \cite[Thm.~5.2.1]{OksendalSDE}) shows that a unique strong solution exists for all times. Up to time $\tau_R$, it clearly coincides with the solution of the SDE Eq.~\eqref{eq:2}. Letting $R\uparrow\infty$ produces a unique strong solution up to the explosion time
\[
  \tau_{\infty}:=\lim\limits_{R\to\infty}\tau_{R}
  =\sup\{t\ge 0:\|X_{t}\|<\infty\}.
\]

\subsection{Lyapunov function}
\label{sec:Lyapunov} 

The following result simplifies the arguments in \cite{leon2018hypoelliptic}.

\begin{lemma}
\label{lem-Lyapunov}
Let \(V(v,w)=1+v^{4}+\lambda w^{2}\geq1\). For $\lambda$ sufficiently large, there exist $\kappa>0$ and $C<\infty$ such that the infinitesimal generator 
\(\mathcal L\) (backward Kolmogorov operator, sometimes also called the Itô generator)  of Eq.~\eqref{eq:2} satisfies
\[
   \mathcal{L}V \le -\kappa V + C\;,
\]
pointwise on $\mathbb{R}^2$. 
\end{lemma}

\begin{proof} Let $x=(v,w)=(x_1,x_2)\in\mathbb{R}^2$. Applying $\mathcal{L}$ to $V$ yields:
\begin{align*}\nonumber
(\mathcal{L}V)(x) &= 
f(x)\partial_vV(x)+g(x)\partial_wV(x)
  + \frac{1}{2}\sum\limits_{i,j=1}^{2} (\sigma^2 NN^{\top})_{ij}\,\partial_{x_i x_j}^{2} V_{\lambda}(x),\\\nonumber
&= [v(a-v)(v-1) - w] \cdot 4v^3 + \varepsilon(bv - cw) \cdot 2\lambda w + \frac{1}{2}\sigma^2\cdot 12v^2, \\
&= -4v^6 + 4(a+1)v^5 - 4a v^4 - 4v^3w + 6\sigma^2 v^2 + 2\lambda\varepsilon b v w - 2\lambda\varepsilon c w^2.
\end{align*}
To prove the claim, it has to be shown that the summands $-4v^6$ and $2\lambda\varepsilon c w^2$ dominate all others for $v$ and $w$ sufficiently large. Clearly $-4v^6$ dominates $4(a+1)v^5 - 4a v^4+ 6\sigma^2 v^2$ for large $v$. For the mixed products $v^3w$ and $vw$, we will use Young's inquality, giving
$$
|v^3w| \leq \frac{v^6}{2} + \frac{w^2}{2} 
\;,
\qquad
|vw| \leq \frac{v^2}{2\eta} + \frac{\eta w^2}{2}
\;,
$$
where $\eta>0$ is a free parameter to be chosen shortly. The first of these two bounds leads to a first contribution $2v^6$, which is smaller than $-4v^6$, and a second contribution which is $2w^2$. This latter has to be compensated by the negative contribution  $- 2\lambda\varepsilon c w^2$. The bound on $|vw|$ gives one summand $2\lambda\varepsilon b \frac{v^2}{2\eta}$ which is again dominated by $v^6$, and a second one $2\lambda\varepsilon b \frac{\eta w^2}{2}$. Choosing $\eta$ sufficiently small, this is dominated by $2\lambda\varepsilon c w^2$.  Finally, choosing $\lambda$ sufficiently large, this last summand indeed dominates $2w^2$. In conclusion, we obtain a bound
$$
(\mathcal{L}V)(x)
\leq -c_v v^6 -c_ww^2 
\leq -c_v v^4 -c_ww^2,
$$
for $v$ and $w$ sufficiently large, for two positive constants $c_v$ and $c_w$. For smaller values of $v$ and $w$, $(\mathcal{L}V)(x)$ is bounded by a constant. Combining these bounds one readily deduces the claim for a $\kappa=\min\{c_v,\frac{c_w}{\lambda}\}$. 
\end{proof}

\subsection{Non-explosion}
\label{sec:NonExplosion} 
Combined with Section~\ref{sec-LocalExist}, the following lemma proves item (i) of Theorem~\ref{thm:FHN-ergodic}.

\begin{lemma}
\label{nonexploseLemma}
$\mathbb{P}(\tau_\infty<\infty)=0$.
\end{lemma}

\begin{proof} Let $V$ be as in Lemma~\ref{lem-Lyapunov} and set
\[
\hat{\tau}_n\;:=\;\inf\bigl\{t\ge0 : V (X_t)\ge n\bigr\},\qquad n\in\mathbb{N} .
\]
By Dynkin’s formula \cite[Thm.~7.4.1]{OksendalSDE} applied to the stopped process,
\begin{equation}\label{eq:Dynkin-stopped}
      \mathbb{E}\bigl[V (X_{t\wedge\hat{\tau}_n})\bigr]
      \;=\;
      V (X_0)
+\mathbb{E}\!\int_{0}^{t\wedge\hat{\tau}_n}\!\mathcal L V (X_s)\,ds ,
      \qquad t\ge0 .
\end{equation}
Hence by Lemma~\ref{lem-Lyapunov} one gets for all $t$
$$
      \mathbb{E}\bigl[V (X_{t\wedge\hat{\tau}_n})\bigr]
      \;\leq\;
      V (X_0)
+Ct
\;.
$$
Taking $n \to \infty$, then $\hat{\tau}_n \to \tau_\infty$ because \(V \) is coercive, namely $V_{\lambda}(x) \to \infty$ as $\|x\| \to \infty$. Suppose now that $\tau_\infty \leq t$ on a set of positive probability. Then $\mathbb{E}\bigl[V (X_{t\wedge\hat{\tau}_n})\bigr]=\infty$, again due to the coercivity of $V$, contradicts the above bound.  Hence \(\mathbb{P}(\tau_\infty\le t)=0\) for every
\(t>0\); \textit{i.e.,}\ \(\tau_\infty=\infty\) almost surely and the solution is
non-explosive in finite time.
\end{proof}

\subsection{H\"ormander property}

\begin{lemma}
\label{hoermanderLemma}
$\mathcal{L}$ and $\mathcal{L}^*$ are hypoelliptic in the sense of H\"ormander.
\end{lemma}

\begin{proof}
Let us define the differential operators:
\[
\mathcal{G}_0 := f \partial_v + g \partial_w, \quad 
\mathcal{G}_1 := \sigma \partial_v,
\]
where $f$ and $g$ are the coefficients of the deterministic dynamics.
The Lie bracket evaluates to: 
\[
[\mathcal{G}_1, \mathcal{G}_0] = \sigma \partial_v(f) \partial_v + \sigma \partial_v(g) \partial_w.
\]
Explicitly:
\[
\partial_v f = -3v^2 + (2a + 2)v - a, \quad 
\partial_v g = \varepsilon b.
\]
Thus:
\[
[\mathcal{G}_1, \mathcal{G}_0] = \sigma (-3v^2 + (2a + 2)v - a) \partial_v + \sigma \varepsilon b \partial_w.
\]

At any \((v,w)\), the vectors 
\[
\mathcal{G}_1 = \sigma \begin{pmatrix} 1 \\ 0 \end{pmatrix}, \quad
[\mathcal{G}_1, \mathcal{G}_0] = \sigma \begin{pmatrix} -3v^2 + (2a + 2)v - a \\ \varepsilon b \end{pmatrix},
\]
are linearly independent because:
\[
\det \begin{pmatrix} \sigma & \sigma (-3v^2 + (2a + 2)v - a) \\ 0 & \sigma \varepsilon b \end{pmatrix} = \sigma^2 \varepsilon b \neq 0.
\]
Therefore \(\text{Lie}(\mathcal{G}_0, \mathcal{G}_1)\) spans \(\mathbb{R}^2\) everywhere. Hence, Hörmander's condition for hypoellipticity  \cite{hormander1967hypoelliptic} holds.
\end{proof}

\subsection{Strong Feller property}
\begin{lemma}\label{lem:strong-feller}
For every $t>0$, the Markov semigroup operator $P_t$ associated with the unique strong
solution $(X_t)_{t\ge 0}=(v_t,w_t)$ of Eq.~\eqref{eq:2}:
\[
P_t\varphi(x)=\mathbb{E}_x\!\bigl[\varphi(X_t)\bigr],
\qquad \varphi\in L^\infty(\mathbb{R}^2),
\]
is strong Feller, namely, the map $x\mapsto P_t\varphi(x)$ is
continuous on $\mathbb{R}^2$. 
\end{lemma}

\begin{proof} 
Because the vector fields
\(\mathcal{G}_1=\sigma\,\partial_v\) and
\([\mathcal{G}_1,\mathcal{G}_0]=\sigma\varepsilon b\,\partial_w\) span \(\mathbb{R}^{2}\) everywhere, Hörmander’s bracket condition holds.
The drift and diffusion coefficients of the infinitesimal generator $\mathcal{L}$ are \(C^{\infty}\);
hence, by Hörmander’s hypoelliptic theorem,
\(P_t(x,u)\) is \(C^{\infty}\) in \((x,u)\) for every fixed \(t>0\). For any \(\varphi\in L^{\infty}(\mathbb{R}^{2})\), the map   \(x\mapsto\int \varphi(u)\,P_t(x,u)\,du\) is continuous, so that for \(t>0\) also
\[
x\mapsto   P_t \varphi(x)
   = \mathbb{E}_{x}\!\bigl[\varphi(X_t)\bigr]
   = \int_{\mathbb{R}^{2}} \varphi(u)\,P_t(x,u)\,du,
\]
is continuous. 
\end{proof}

\subsection{Topological irreducibility}
\begin{lemma}
\label{lem:irreducibility}
Fix $T > 0$, a target $X^* = (v^*, w^*) \in \mathbb{R}^2$, and $r > 0$. For every initial state $x = (v_0, w_0) \in \mathbb{R}^2$, the solution $X_t = (v_t, w_t)$ of Eq. \eqref{eq:2} satisfies
\begin{equation}\label{eq:irr-goal}
\mathbb{P}_x \big( X_T \in B_r(X^*) \big) = \eta(r, T, x) > 0.
\end{equation}
\end{lemma}

Two different proofs can be found in the literature: one is based on Girsanov's formula \cite{Wu2001} and is given in  \cite[Proposition 2.2]{leon2018hypoelliptic}, and another one on endpoint control and Strook's theorem, see \cite[Lemma 3.4]{mattingly2002ergodicity}.

\subsection{Conclusion of the proof of Theorem~\ref{thm:FHN-ergodic}}
\begin{proof} (Theorem~\ref{thm:FHN-ergodic}) Item  (i)  is already proved in Lemma~\ref{nonexploseLemma}. Items (ii) and (iii) follow from \cite[Theorem 5.2c]{down1995exponential} (just as discussed in \cite[Theorem 2.4]{Wu2001}) and the fact that FHN process in Eq.\eqref{eq:2} is a strong Feller semigroup (Lemma~\ref{lem:strong-feller}) which is topologically transitive (Lemma~\ref{lem:irreducibility}) and the generator of which satisfies 
\[
\mathcal{L}V \le -\kappa V + C,
\]
where $V$ is the Lyapunov function and $\kappa$ and $C$ are positive constants (Lemma~\ref{lem-Lyapunov}). As $V$ increases as $x=(v,w)\to\infty$, this indeed implies the condition (C2) in \cite{Wu2001}. One concludes that the invariant measure is unique. Due to hypoellipticity, it is absolutely continuous with a smooth density. The positivity of its density follows from Bony's maximum principle, or alternatively from the topological transitivity mentioned above. Item  (iv)  follows directly from item (iii).
\end{proof}

\section*{Data availability}
The data generated during the current study are available from the corresponding author upon reasonable request.


\end{document}